\newif\ifanonym
\anonymfalse %last one wins (comment this line as needed)

\documentclass[11pt,letterpaper]{article}
\usepackage[margin=1in]{geometry}

\usepackage{amsmath,amsthm,amssymb,amsfonts}
\usepackage{xspace}
\usepackage{url}
\usepackage{hyperref}
\usepackage[dvipsnames]{xcolor}
\usepackage[ruled,linesnumbered,noend]{algorithm2e}
\usepackage{appendix}
\usepackage{babel}

\usepackage{graphicx} % Required for inserting images
\usepackage{booktabs} % For formal tables
\usepackage[ruled]{algorithm2e} % For algorithms
\usepackage{subfigure}
\usepackage{caption}

\usepackage{multirow}
\usepackage{afterpage}
\usepackage{lipsum}
\usepackage{longtable}
\usepackage{booktabs}
\usepackage{thm-restate}
\usepackage{graphicx}
\usepackage{subcaption}
\usepackage[normalem]{ulem}
\usepackage{bm}
\usepackage{enumerate}
\usepackage{bbm}
\usepackage{mathtools}
\usepackage{mathrsfs}
\usepackage{minted}
\usepackage[capitalise,noabbrev]{cleveref}
\usepackage{listings}
\usepackage{comment}

\lstdefinestyle{cppstyle}{
  language=C++,
  basicstyle=\ttfamily\small,
  keywordstyle=\color{blue}\bfseries,
  commentstyle=\color{gray},
  stringstyle=\color{orange!70!black},
  numbers=left,
  numberstyle=\tiny\color{gray},
  stepnumber=1,
  numbersep=8pt,
  tabsize=2,
  showspaces=false,
  showstringspaces=false,
  breaklines=true,
  frame=single,
  rulecolor=\color{black!30},
  backgroundcolor=\color{gray!5},
  upquote=true,
  literate={"}{{"}}1
}

\hypersetup{colorlinks=true,allcolors=blue}

\allowdisplaybreaks[4]

\makeatletter
\newcounter{HALG@line}
\renewcommand{\theHALG@line}{\thealgorithm.\arabic{ALG@line}}
\makeatother

\usepackage[obeyFinal,textsize=footnotesize]{todonotes}

\let\epsilon\varepsilon

\usepackage{amsmath}
\usepackage{amsthm}
\usepackage{amssymb}
\usepackage{comment}
\usepackage{booktabs} % For formal tables
\usepackage[numbers]{natbib}
\usepackage{bm}
\usepackage{tikz}
\usetikzlibrary{arrows.meta, positioning}
\usepackage{rotating}
\usepackage{multicol}
\usepackage{multirow} 
\usepackage{enumitem}
\usepackage{hyperref}[backref=page]
\usepackage{caption}
\usepackage{graphicx}
\usepackage{subfigure}

\newtheorem{theorem}{Theorem}[section]

\newtheorem{proposition}[theorem]{Proposition}
\newtheorem{lemma}[theorem]{Lemma}%
\newtheorem{corollary}[theorem]{Corollary}%
\newtheorem{example}[theorem]{Example}

\newtheorem{definition}[theorem]{Definition}
\newtheorem{claim}[theorem]{Claim}

\newcommand{\bv}{\begin{array}}

\newcommand{\Omit}[1]{}

\newcommand{\Mwin}{\mathcal{M}_k}
\newcommand{\Mlose}{\mathcal{M}_{\ell}}

\newcommand{\topp}{t_{\ell}}
\newcommand{\ntopp}{n_{\ell}}
\newcommand{\utopp}{u_{\ell}}

\DeclareMathOperator{\CH}{CH}

\begin{document}

\title{Likelihood of the Existence of Average Justified Representation}
\ifanonym
\author{Anonymous Authors}
\else
\author{%
  Qishen Han \\
  Rutgers University \\
  \texttt{hnickc2017@gmail.com} \\
  \and
  Biaoshuai Tao \\
  Shanghai Jiao Tong University \\
  \texttt{bstao@sjtu.edu.cn} \\
  \and
  Lirong Xia \\
  Rutgers University \\
  \texttt{xialirong@gmail.com} \\
  \and
  Chengkai Zhang \\
  Rutgers University \\
  \texttt{chengkai.zhang@rutgers.edu} \\
  \and
  Houyu Zhou \\
  UNSW Sydney \\
  \texttt{houyu.zhou@unsw.edu.au} \\
}
\fi
\date{}
\maketitle

\begin{abstract}
We study the approval-based multi-winner election problem where $n$ voters jointly decide a committee of $k$ winners from $m$ candidates. We focus on the axiom \emph{average justified representation} (AJR) proposed by Fern{\'{a}}ndez, Elkind, Lackner, Garc{\'{\i}}a, Arias{-}Fisteus, Basanta{-}Val, and Skowron (2017). AJR postulates that every group of voters with a common preference should be sufficiently represented in that their average satisfaction should be no less than their Hare quota. Formally, for every group of $\lceil\ell\cdot\frac{n}{k}\rceil$ voters with $\ell$ common approved candidates, the average number of approved winners for this group should be at least $\ell$. It is well-known that a winning committee satisfying AJR is not guaranteed to exist for all multi-winner election instances. In this paper, we study the likelihood of the existence of AJR under the Erd\H{o}s--R\'enyi model. We consider the Erd\H{o}s--R\'enyi model parameterized by $p\in[0,1]$ that samples multi-winner election instances from the distribution where each voter approves each candidate with probability $p$ (and the events that voters approve candidates are independent), and we provide a clean and complete characterization of the existence of AJR committees in the case where $m$ is a constant and $n$ tends to infinity. We show that there are two phase transition points $p_1$ and $p_2$ (with $p_1\leq p_2$) for the parameter $p$ such that: 1) when $p<p_1$ or $p>p_2$, an AJR committee exists with probability $1-o(1)$, 2) when $p_1<p<p_2$, an AJR committee exists with probability $o(1)$, and 3) when $p=p_1$ or $p=p_2$, the probability that an AJR committee exists is bounded away from both $0$ and $1$.
\end{abstract}

\section{Introduction}
\label{sec:intro}
In an approval-based multi-winner election, $n$ voters jointly decide a committee of $k$ winners from $m$ candidates by casting approval ballots (i.e., each voter's ballot contains a binary string of length $m$ indicating whether or not she approves each candidate).
The set of $k$ winners is called the winning committee.
Approval-based multi-winner election has been extensively studied in the past literature (see the recent book by \citet{lackner2023multiwinner} for a comprehensive survey).
Fairness in multi-winner approval voting is crucial for democratic governance, ensuring that diverse voter preferences are adequately reflected in elected committees. 
Among those fairness notions, \emph{proportional representation} is arguably the most important consideration, and it postulates that every party of the voters should receive a number of seats proportional to its population.
For example, to proportionally represent a party of $n/2$ voters, at least $k/2$ candidates approved by them should be selected as winners.
For elections where voters are not pre-specified into parties/groups, recent work has been considering \emph{justified representation (JR)} and its variants \cite{aziz2017jr} which generalize the concept of proportional representation.
Along this line of fairness axioms, the notion of ``party'' is naturally generalized to the notion of \emph{cohesive group}---a group of voters with similar preferences reflected by their approval ballots, and the winning committee should ensure that every cohesive group that justifies $\ell$ \emph{Hare quotas} should somehow have a satisfaction level of $\ell$ in some sense. Here, the Hare quota~\cite{pukelsheim2017quota}, first proposed by Alexander Hamilton for use in United States congressional apportionment, is defined as $n/k$---the minimum number of voters in a group that justifies one seat in the winning committee. 

Formally, we say that a set of voters $V$ is an \emph{$\ell$-cohesive group} if $|V|\geq\ell\cdot\frac nk$ and voters in $V$ approve at least $\ell$ candidates in common.
To formulate what it means by saying an $\ell$-cohesive group has a satisfaction level of $\ell$, different criteria have been proposed that yield different variants of JR.
Naturally, the satisfaction of a single voter can be defined as the number of candidates in the winning committee that she approves.
It is then natural to define the satisfaction of an $\ell$-cohesive group to be the average satisfaction of the voters in this group.
As an $\ell$-cohesive group deserves $\ell$ quotas in the spirit of proportional/justified representation, we would like to require that \textbf{every $\ell$-cohesive group has the average satisfaction at least $\ell$}.
This is exactly the notion of \emph{average justified representation} (AJR) proposed by~\citet{fernandez2017pjr}.
%In a one-sentence definition, a winning committee satisfies AJR if the average satisfaction of every $\ell$-cohesive group is at least $\ell$.
The formal mathematical definition of AJR is deferred to Definition~\ref{def:AJR}.

While AJR is a natural axiom aligning with the concept of proportional/justified representation, a winning committee satisfying AJR may not exist, as demonstrated in the example in Fig.~\ref{fig:AJR_not_exist}.
Subsequent work by \citet{aziz2018ejr,skowron2021proportionality} further explored this notion, showing that the well-known \emph{Proportional Approval Voting} (PAV) rule guarantees an average satisfaction of $\ell-1$ for every $\ell$-cohesive group,\footnote{More precisely, PAV guarantees an average satisfaction of $\ell - 1 + o(1/k)$. Moreover, for every constant $c > 0$, no rule can guarantee average satisfaction at least $(\ell - 1 + c)$ in general. Hence, we refer to $\ell - 1$ as its average satisfaction for simplicity.} one below the proportional ideal of $\ell$.

In addition, those worst-case scenarios where AJR committees fail to exist may be rare in practice. 
Indeed, empirical evidence suggests that AJR might often hold: \citet{brill2022individual,brill2025individual} studies a stronger representation axiom (requiring every voter in an $\ell$-cohesive group to individually approve at least $\ell$ winners). 
Their experiments, across multiple probabilistic models, show that for large enough committee sizes $k$, there exist winning committees satisfying this stronger condition (that implies AJR). 
This disparity between worst-case impossibility result and empirical results motivates a theoretical investigation of AJR under probabilistic assumptions.

\begin{figure}[!htb]
    \centering
    \begin{tikzpicture}[scale=0.4]
        \definecolor{fancyblue}{RGB}{127,172,204}  
        \definecolor{fancyyellow}{RGB}{246,189,78}    
        \definecolor{fancyred}{RGB}{233,108,102}  
        \definecolor{fancygreen}{RGB}{31,145,158}    
        \draw[fancyblue, line width=0.5mm, rounded corners] (0.2,0) rectangle ++(1.6,8);
        \draw[fancyyellow, line width=0.5mm, rounded corners] (0,0.2) rectangle ++(8,1.6);
        \draw[fancyred, line width=0.5mm, rounded corners] (0,6.2) rectangle ++(8,1.6);
        \draw[fancygreen, line width=0.5mm, rounded corners] (6.2,0) rectangle ++(1.6,8);
        \draw[black, line width=0.25mm] (1,7) circle (0.6);
        \draw[black, line width=0.25mm] (3,7) circle (0.6);
        \draw[black, line width=0.25mm] (5,7) circle (0.6);
        \draw[black, line width=0.25mm] (7,7) circle (0.6);
        \draw[black, line width=0.25mm] (1,1) circle (0.6);
        \draw[black, line width=0.25mm] (3,1) circle (0.6);
        \draw[black, line width=0.25mm] (5,1) circle (0.6);
        \draw[black, line width=0.25mm] (7,1) circle (0.6);
        \draw[black, line width=0.25mm] (1,3) circle (0.6);
        \draw[black, line width=0.25mm] (1,5) circle (0.6);
        \draw[black, line width=0.25mm] (7,3) circle (0.6);
        \draw[black, line width=0.25mm] (7,5) circle (0.6);
    \end{tikzpicture}
    \caption{An example where AJR committees do not exist. In this example, each black circle stands for a voter. The colored rectangles stand for candidates. If a circle is in a rectangle, it means that this voter approves this candidate. In this example, $n=12$ and $m=4$, and we are to select a winning committee of size three (i.e., $k=3$). In this case, there are four $1$-cohesive groups, each corresponding to one side of the large square. Selecting any three of these four candidates will lead to some group's average satisfaction being only $0.5$, which is less than the requirement $1$.}
    \label{fig:AJR_not_exist}
\end{figure}
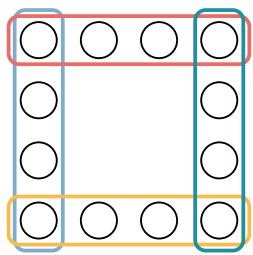

We initiate the study of the likelihood of AJR by using a simple, well-established model in multi-winner voting: \emph{the Erd\H{o}s--R\'enyi bipartite random graph} where each voter approves each candidate independently with probability $p$ (see, e.g., \citet{PetersP0021,SzufaFJLSST22,brill2022individual,elkind2023justifying}). 
This model captures a neutral ``base case'' where voter preferences lack inherent structure—a natural starting point for theoretical analysis.
We study the following problem in this paper.

\begin{quote}
    \emph{Under what values of $p$ does an AJR committee exist with high probability? Conversely, when is an AJR committee unlikely to exist?}
\end{quote}

Motivated by that the number of candidates $m$ is typically small while the number of voters $n$ can be very large (e.g., a political election in a country), we assume $m$ is a constant (and so is the size of the winning committee $k$) and study the above problem by analyzing the probabilities asymptotically with $n\rightarrow\infty$.
Under the Erd\H{o}s--R\'enyi model, AJR holds trivially at the extremes: when $p=1$, all committees satisfy AJR; when $p=0$, no cohesive group exists, so every winning committee automatically satisfies AJR. 
Naturally, we may expect that AJR committees exist with probability $1-o(1)$ for small enough $p$ (in which case no cohesive group exists with high probability) and large enough $p$ (in which case, with high probability, every cohesive group has a high enough satisfaction for an arbitrary winning committee).
The core challenge lies in the intermediate regime: is there any intermediate regime of $p$ where AJR committees are unlikely to exist? 
Does the likelihood for the existence of AJR committees change smoothly as $p$ changes? Or, are there any phase transition points? If so, how many?

\paragraph{On Other Justified Representation Variants.}
Many variants along the line of justified representation that are weaker than AJR have been proposed, with the most notable examples of \emph{JR}~\cite{aziz2017jr}, \emph{extended JR} (EJR)~\cite{aziz2018ejr}, and \emph{proportional JR} (PJR)~\cite{fernandez2017pjr,brill2017phragmen,brill2024phragmen}.
For EJR, it is required that every $\ell$-cohesive group contains \emph{one} voter with a satisfaction of at least $\ell$.
JR is an even weaker notion than EJR with the same requirement, except that only $1$-cohesive groups are concerned.
PJR is also a weaker notion than EJR that requires each $\ell$-cohesive group must have a \emph{total} satisfaction of at least $\ell$ (or equivalently, an average satisfaction of $\ell/|V|$ for every $\ell$-cohesive group $V$).
All these three axioms are guaranteed to be satisfiable, and, for each of them, a winning committee satisfying the axiom can be found in polynomial time~\cite{aziz2017jr,aziz2018complexity,peters2021proportional}.
It is easy to see that EJR (and thus JR) is substantially weaker than AJR: it only ensures the satisfaction of \emph{one} voter in each $\ell$-cohesive group (instead of the \emph{average} satisfaction of the whole group) being at least $\ell$. 
PJR, on the other hand, places requirements with respect to average satisfactions, but it is even weaker than EJR, as mentioned above.
In Example~\ref{example:EJR}, we demonstrate that the requirement postulated by EJR (and thus PJR and JR) can be very lenient that admits both natural winning committees and not-so-natural committees, while the requirement of AJR refines the set of valid winning committees by only keeping those natural ones.

\begin{example}
\label{example:EJR}
    Consider a voting instance with $n=k^2$ and $m=2k$ where $k$ is the size of the winning committee.
    The voters are indexed by $\{v_{ij}\}_{i=1,\ldots,k;j=1,\ldots,k}$.
    There are $k$ ``row'' candidates $r_1,\ldots,r_k$ and $k$ ``column'' candidates $c_1,\ldots,c_k$.
    Each voter $v_{ij}$ approves exactly two candidates: $r_i$ and $c_j$.

    In this instance, each of the $2k$ candidates ``induces'' a $1$-cohesive group, i.e., each ``row'' and each ``column'' of $\{v_{ij}\}_{i=1,\ldots,k;j=1,\ldots,k}$ form a $1$-cohesive group.
    It is easy to see that any set of $k$ candidates gives an EJR (and thus PJR and JR) committee.
    In this case, the satisfactions for some of the voters can be $0$.
    On the other hand, there are only two AJR committees: $\{r_1,\ldots,r_k\}$ and $\{c_1,\ldots,c_k\}$.
    In any of the two AJR committees, every voter's satisfaction is $1$.
\end{example}

In addition to those JR variants mentioned above that are weaker than AJR, other JR variants that are incomparable to AJR have also been proposed.
For example, \citet{peters2021proportional} introduced full JR (FJR), which considers not only cohesive groups but also near-cohesive ones. Building on this idea, \citet{KalayciLK2025} extended PJR to that setting and proposed full PJR (FPJR). \citet{brill2023ejr+} introduced EJR+ and PJR+, which offer an alternative weakening of cohesiveness.
The relationship between these notions is given in Fig.~\ref{fig:jr}.

All these variants of JR mentioned above are \emph{yes-or-no} measurements, i.e., a winning committee either satisfies the requirement or not.
Previous literature has also proposed many \emph{quantitative} measurements on the fairness of a winning committee, such as \emph{proportionality degree}~\cite{aziz2018ejr, skowron2021proportionality,janeczko2022proportionality} and \emph{EJR degree}~\cite{TaoZZ24}.
Roughly speaking, given a winning committee, its proportionality degree is the average satisfaction for the cohesive group with the least average satisfaction (a more formal definition is in the next paragraph), and its EJR degree is the minimum (taking over all cohesive groups) number of voters in an $\ell$-cohesive group whose satisfactions are at least $\ell$.
The notion of EJR degree is a quantitative generalization of EJR: instead of just \emph{requiring} at least \emph{one} voter in an $\ell$-cohesive group with satisfaction at least $\ell$, this notion further investigates if ``one'' is the best number possible and tries to \emph{optimize} this number.

The notion of proportionality degree, on the other hand, is more aligned with the notion of AJR by putting the average satisfaction as an objective to be maximized.
More formally, given a function $f:\mathbb{Z}^+\to\mathbb{R}_{\geq0}$, the proportionality degree of a winning committee is $f$ if the average satisfaction of each $\ell$-cohesive group is at least $f(\ell)$.
In terms of this, AJR exactly requires a proportionality degree of $f(\ell)=\ell$, which, as we have seen in Fig.~\ref{fig:AJR_not_exist}, may not always be satisfiable.
As we have also mentioned before, the PAV rule can guarantee a proportionality degree of $f(\ell)=\ell-1$.
Nevertheless, we would also like to remark that the particular choice $f(\ell)=\ell$ as the requirement in the definition of AJR has its special meaning: it faithfully reflects the idea of proportional representation as $\ell$ is exactly the Hare quota for an $\ell$-cohesive group.

\paragraph{Core Stability.}
Along the line of defining fairness based on proportional representation, another criterion is \emph{core stability}~\cite{gillies1959solutions,shapley1969market} which has also been studied widely~\cite{JiangMW20,MunagalaSWW22,MavrovMS23,Xia2025linear}.
In the context of multi-winner election, a winning committee $W$ of size $k$ is \emph{core stable} if there is no \emph{blocking coalition}, where a subset $V$ of voters is said to form a blocking coalition if there is another winning committee $W'$ of size $\lfloor k\frac{|V|}n\rfloor$ such that every voter in $V$ has a strictly higher satisfaction for $W'$ compared to $W$.
Consider a group of voters $V$ with size $|V|=\ell\cdot\frac nk$.
Instead of imposing a constraint $\ell$ on the satisfaction of $V$ like in those JR notions, core stability requires that there is no way to select $\ell$ candidates such that every voter in $V$ is strictly happier compared to the current winning committee $W$.
Notice also that there is no requirement that those voters in $V$ must approve $\ell$ common candidates (i.e., be an $\ell$-cohesive group).
Unlike the case with AJR that is known to be unsatisfiable for some instances, it is an open problem if a core stable committee always exists.

AJR and core stability do not imply each other (see Fig.~\ref{fig:jr} where we have also included core stability).
The example in Fig.~\ref{fig:AJR_not_exist} shows that a core stable committee may not satisfy AJR: it is easy to verify that every set of $3$ candidates forms a core stable winning committee, while we have seen that no AJR committee exists in this instance.
To see that an AJR commitee may not be core stable, consider an instance with $n=8$, $k=4$, and candidates $\{x_1,x_2,x_3,x_4,y,z\}$.
Voters 1, 2, and 3 approve $\{x_1,y\}$, voters 4, 5, and 6 approve $\{x_1,z\}$, and voters 7 and 8 approve $\{x_2,x_3,x_4\}$.
Committee $W=\{x_1,x_2,x_3,x_4\}$ satisfies AJR.
However, it does not satisfy core stability, as the set of voters $\{1,2,3,4,5,6\}$ can deviate to $W'=\{x_1,y,z\}$.

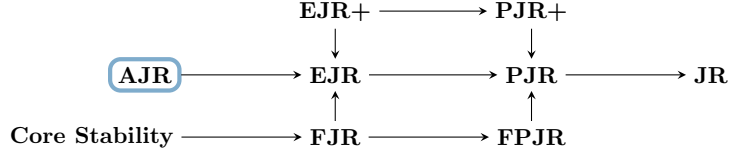
\begin{figure}[!htb]
        \centering
        \begin{tikzpicture}[
        scale=0.8,
        transform shape,
        node distance=0.5cm and 2cm,
        every node/.style={font=\bfseries},
        ->, >=stealth]
        \definecolor{fancyblue}{RGB}{127,172,204}  
        \definecolor{fancyyellow}{RGB}{246,189,78}    
        \definecolor{fancyred}{RGB}{233,108,102}  
        \definecolor{fancygreen}{RGB}{31,145,158}    
            % Nodes
            \node (CS) {Core Stability};
            \node[right=of CS] (FJR) {FJR};
            \node[above=of FJR] (EJR) {EJR};
            \node[right=of FJR] (FPJR) {FPJR};
            \node[above=of FPJR] (PJR) {PJR};
            \node[right=of PJR] (JR) {JR};
            \node[left=of EJR, draw=fancyblue, ultra thick, rounded corners] (AJR) {AJR};
            %\node[left=of AJR] (IR) {IR};
            \node[above=of EJR] (EJRplus) {EJR+};
            \node[above=of PJR] (PJRplus) {PJR+};
            % Arrows
            \draw (CS) -> (FJR);
            \draw (FJR) -> (FPJR);
            \draw (FJR) -> (EJR);
            \draw (FPJR) -> (PJR);
            \draw (EJR) -> (PJR);
            \draw (PJR) -> (JR);
            \draw (EJRplus) -> (EJR);
            \draw (EJRplus) -> (PJRplus);
            \draw (PJRplus) -> (PJR);
            \draw (AJR) -> (EJR);
            %\draw (IR) -> (AJR);
        \end{tikzpicture}
        \caption{Diagram of relationships among different variants of JR. Arrows indicate implications.}
        \label{fig:jr}
\end{figure}

\section{Overview of Contributions}

%\subsection{Our Model}
For an integer $m$, we use $[m]$ to denote the set $\{1, 2, \ldots, m\}$.

\paragraph{Approval-Based Multi-Winner Election.} There are $n$ voters who cast approval votes among a group of candidates $M = [m]$. We use $i$ to denote a generic candidate and $v$ to denote a generic voter. Each voter can approve an arbitrary number of candidates. For a voter $v$, we use $A(v)$ to denote the set of candidates approved by $v$, i.e., approval ballot of voter $v$.  
The goal is to elect a winner {\em committee} of size $k$, denoted by $W$. 
% Given a group of voters $V$ and a committee $W$, let $U(V,W)$ be a vector where for each $v\in V$, the entry $u(v,W)=|A(v)\cap W|$. Let $U_{\max}(V,W)=\max_{v\in V} u(v,W)$ be the maximum value among all entries in $U(V,W)$. Let $U_{\min}(V,W)=\min_{v\in V} u(v,W)$ be the minimum value among all entries in $U(V,W)$. Let $U_{\av}(V,W)=\frac{1}{|V|}\sum_{v\in V}u(v,W)$ be the average value of all entries in $U(V,W)$.

\begin{definition}[Average Justified Representation (AJR)] \label{def:AJR}
A winning committee $W$ of size $|W| = k$ provides AJR if for every $\ell \in [k]$ there does not exist a group of voters $V$ satisfying 
\begin{enumerate}
    \item Size Constraint: $|V| \ge \ell \cdot \frac{n}{k}$;
    \item Cohesiveness Constraint: $|\bigcap_{v\in V} A(v)| \ge \ell$;
    \item Underrepresented Constraint: $\sum_{v\in V} |A(v) \cap W| <\ell\cdot |V|$. 
\end{enumerate}
\end{definition}

For each $\ell$, we call a set/group of voters $V$ an \emph{$\ell$-cohesive group towards} $L$ if $|V| \ge \ell\cdot \frac{n}{k}$, $L \subseteq |\bigcap_{v\in V} A(v)|$, and $|L| = \ell$ (corresponding to the first two constraints in the definition above). We call an $\ell$-cohesive group $V$ \emph{underrepresented} if it satisfies Constraint (3) (in addition to satisfying (1) and (2)).  Let $U_v(W) = |A(v) \cap W|$ be the \emph{utility} of voter $v$ achieved by the winning committee $W$, $U(V,W) = \sum_{v\in V} |A(v) \cap W|$ be the total utility of $V$ achieved by $W$, and $u(U, V) = \frac{1}{|V|} U(V, W)$ be the average utility of $V$ achieved by $W$.
Then, an $\ell$-cohesive group $V$ is underrepresented if the average utility of $V$ is less than $\ell$.

\paragraph{Erd\H{o}s--R\'enyi Bipartite Model.} 
The approval ballot $A(v)$ of each voter $v$ is generated i.i.d.. A voter $v$ has a probability of $p \in [0, 1]$ to approve each candidate $i \in M$. The events of $v$ approving different candidates are independent. Therefore, the probability that $v$'s approval ballot is $A(v)$ is $p^{|A(v)|} \cdot (1 - p)^{m - |A(v)|}$. We use $\pi$ to denote the distribution of $A(v)$.

%In this paper, we look into the likelihood that an AJR committee $W$ exists when the approval ballot of each voter follows $\pi$. 

\subsection{Our Results}
\label{sect:ourresults}
We fully characterize the existence of an AJR committee under the Erd\H{o}s--R\'enyi bipartite model. This paper presents the following result.

\begin{theorem}\label{thm:main}
    For any constant $m$, $k$, and $p$ with $m> k\geq 2$ and $p\in[0,1]$, and for $n\to\infty$, the likelihood that there exists a committee $W$ that provides AJR has the following trichotomy.
    Let $p_1^\ast=\frac1k$ and $p_2^\ast$ be the maximum $x\in[0,1]$ such that
    \begin{equation}\label{eqn:thmmain}
        k(2x(1-x)^k+kx^2(1-x)^{k-1})=1.
    \end{equation}
    Then the likelihood that an AJR committee exists is
    \begin{enumerate}
        \item $1-o(1)$ if $p<p_1^\ast$ or $p>p_2^\ast$;
        \item $o(1)$ if $p_1^\ast<p<p_2^\ast$;
        \item $\Theta(1)$ and $1-\Theta(1)$ (i.e., bounded away from $0$ and $1$) if $p=p_1^\ast$ or $p=p_2^\ast$.
    \end{enumerate}
\end{theorem}

As the first remark, the case with $k=1$, which is excluded from Theorem~\ref{thm:main}, is trivial: an AJR committee always exists.
To see this, in the case $k=1$, the only possible cohesive group is the set of all voters.
If they approve a common candidate, then selecting this candidate to be the winning committee gives an AJR committee.
Otherwise, there is no cohesive group at all, in which case an arbitrary winning committee automatically achieves the AJR requirement.

Theorem~\ref{thm:main} captures all the remaining cases.
The case for $k=2$ is also special, Equation~(\ref{eqn:thmmain}) only has one solution $x=\frac12$, in which case $p_1^\ast=p_2^\ast$.
In this case, case (2) of the theorem never applies, and case (3) applies only when $p=\frac12$.
For general cases $k\geq3$, the following proposition shows that $p_2^\ast>p_1^\ast=\frac1k$ and $p_2^\ast$ is of the same order of $\frac1k$.
Notice that the upper bound $\frac5k$ is not tight, but it shows that $p_2^\ast\rightarrow 0$ as $k\rightarrow\infty$.

\begin{proposition}\label{prop:p2}
    For $k\geq3$, we have $\frac1k<p_2^\ast<\min\{1,\frac5k\}$ for $p_2^\ast$ defined in Theorem~\ref{thm:main}.
\end{proposition}
The proof of this proposition uses some observations in the proof of Theorem~\ref{thm:main}, and is deferred to Sect.~\ref{sect:p2}.

Therefore, for $k\geq 3$, Theorem~\ref{thm:main} identifies two phase-transition points for $p$, namely, $p_1^\ast$ and $p_2^\ast$.
For $p<p_1^\ast$ and $p>p_2^\ast$, an AJR committee exists with high probability.
For $p_1^\ast<p<p_2^\ast$, no AJR committee exists with high probability.
At the two phase-transition points, the probability for the existence of an AJR committee is bounded away from $0$ and $1$.
It does not come as a surprise that an AJR committee exists with high probability for $p$ being either too small or too large.
Intuitively, for very small $p$, it is likely that there is no cohesive group at all, so every winning committee automatically satisfies AJR; for very large $p$, every candidate is approved by a large fraction of voters, and the AJR condition is satisfied for every committee of $k$ candidates.
It is much less intuitively clear what the story is for $p$ being neither very small nor very large, and our Theorem~\ref{thm:main} states that the AJR committee is unlikely to exist for $p$ falling into a \emph{single} intermediate interval.
In addition, the length of this intermediate interval decreases to $0$ (see Proposition~\ref{prop:p2}) as $k\rightarrow\infty$.

In Fig.~\ref{fig:p1p2}, we plot the values of $p_1^\ast$ and $p_2^\ast$ for $k=2,\ldots,10$.

\begin{figure}[!htb]
    \centering
    \begin{tikzpicture}
        \draw[->] (0,0)--(6,0);
        \draw[->] (0,0)--(0,6);
        \draw[dashed] (0,1)--(5.5,1);
        \draw[dashed] (0,1.5)--(5.5,1.5);
        \draw[dashed] (0,2)--(5.5,2);
        \draw[dashed] (0,2.5)--(5.5,2.5);
        \draw[dashed] (0,3)--(5.5,3);
        \draw[dashed] (0,3.5)--(5.5,3.5);
        \draw[dashed] (0,4)--(5.5,4);
        \draw[dashed] (0,4.5)--(5.5,4.5);
        \draw[dashed] (0,5)--(5.5,5);
        \filldraw (0,0) circle (1pt);
        \filldraw (5,0) circle (1pt);
        \node[anchor=east] at (0,1) {$k=2$};
        \node[anchor=east] at (0,1.5) {$k=3$};
        \node[anchor=east] at (0,2) {$k=4$};
        \node[anchor=east] at (0,2.5) {$k=5$};
        \node[anchor=east] at (0,3) {$k=6$};
        \node[anchor=east] at (0,3.5) {$k=7$};
        \node[anchor=east] at (0,4) {$k=8$};
        \node[anchor=east] at (0,4.5) {$k=9$};
        \node[anchor=east] at (0,5) {$k=10$};
        \node[anchor=north] at (0,0) {$0$};
        \node[anchor=north] at (5,0) {$0.5$};
        \node[anchor=north] at (6,0) {$p$};
        \filldraw (10/2,1) circle (2pt);
        \filldraw (10/3,1.5) circle (2pt);
        \filldraw (10/4,2) circle (2pt);
        \filldraw (10/5,2.5) circle (2pt);
        \filldraw (10/6,3) circle (2pt);
        \filldraw (10/7,3.5) circle (2pt);
        \filldraw (10/8,4) circle (2pt);
        \filldraw (10/9,4.5) circle (2pt);
        \filldraw (10/10,5) circle (2pt);
        \draw (10/2,1)--(10/3,1.5)--(10/4,2)--(10/5,2.5)--(10/6,3)--(10/7,3.5)--(10/8,4)--(10/9,4.5)--(10/10,5);
        \filldraw (5,1) circle (2pt);
        \filldraw (4.51333,1.5) circle (2pt);
        \filldraw (3.8,2) circle (2pt);
        \filldraw (10*0.327,2.5) circle (2pt);
        \filldraw (10*0.286667,3) circle (2pt);
        \filldraw (10*0.254857,3.5) circle (2pt);
        \filldraw (10*0.23,4) circle (2pt);
        \filldraw (10*0.209111,4.5) circle (2pt);
        \filldraw (10*0.192,5) circle (2pt);
        \draw (5,1)--(4.51333,1.5)--(3.8,2)--(10*0.327,2.5)--(10*0.286667,3)--(10*0.254857,3.5)--(10*0.23,4)--(10*0.209111,4.5)--(10*0.192,5);
        \node[anchor=south] at (10/10,5.3) {$p_1^\ast$};
        \node[anchor=south] at (10*0.192,5.3) {$p_2^\ast$};
    \end{tikzpicture}
    \caption{The values of $p_1^\ast$ and $p_2^\ast$ for $k=2,\ldots,10$.}
    \label{fig:p1p2}
\end{figure}
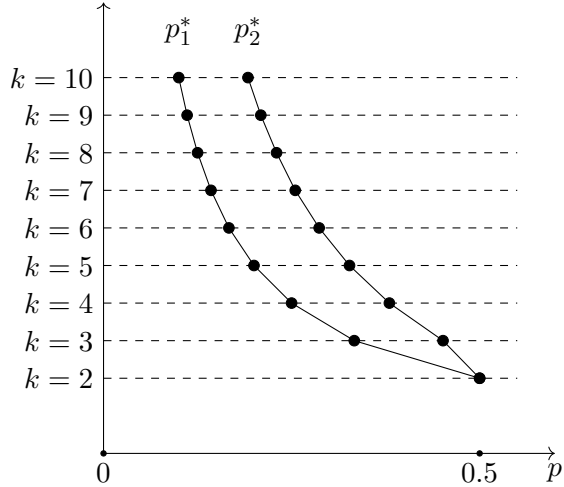

\subsection{Comparisons with Other Work}
\label{sect:compare}
Our result theoretically verifies the empirical findings in~\citet{brill2022individual,brill2025individual}: for large $k$, most values of $p$ would guarantee the existence (with high probability) of AJR committees.
Moreover, our result provides a fine-grained characterization of the existence of AJR committees by precisely describing the relationship between $p$ and the likelihood of AJR committees.

Empirical studies for EJR, PJR, and JR (discussed in Sect.~\ref{sec:intro}) have also been performed, with the main focus on the \emph{number} of valid committees (since we know valid committees always exist for these notions).
For EJR and PJR, under the Erd\H{o}s--R\'enyi model, \citet{bredereck2019experimental} shows that the number of valid committees goes down and then up as $p$ increases, and, for JR, \citet{elkind2023justifying} shows that this number is constantly $\binom{m}{k}$ (i.e., all committees are valid) for $n\rightarrow\infty$.
Our theoretical result is similar in that all committees satisfy AJR for small or large $p$.
The difference is that, for values of $p$ that are neither too small nor too large, \citet{bredereck2019experimental} shows that the number of EJR committees goes down, while ours shows that no AJR committee exists at all.
This makes a distinction between AJR and EJR/PJR/JR.
Notice also that both \citet{bredereck2019experimental} and \citet{elkind2023justifying} also consider the setting where fewer than $k$ candidates are selected and whether those JR notions are still satisfiable, which is not the focus of this paper.

Theoretically, \citet{Xia2025linear} studies the similar problem on the likelihood of existence of committees satisfying these axioms. 
Our result for AJR makes a sharp contrast with Xia's result on core stability.
In Xia's setting, each candidate $c_j$ is independently approved with probability $p_j$, and Xia finds that a committee satisfying core stability is very likely to exist regardless of the choice of $p_j$. 
In our work, we consider a more specialized model in which each candidate is approved independently with the same probability $p$. 
Despite that our model is the symmetric special case, we show that there is a regime where AJR does not exist with high probability, suggesting that AJR may be harder to satisfy than CS.
For example, for some natural parameters such as $m=10$, $k=4$, and $p=1/3$, for sufficiently large $n$, almost every committee satisfies core stability, while none of them satisfies AJR.

\subsection{Structure of Our Paper}
The remaining part of our paper is organized as follows.
In Sect.~\ref{sect:lemproof}, we provide an intermediate characterization (Lemma~\ref{lem:main}) that characterizes the phase-transition points for $p$ for the existence of AJR committees. The characterization is cumbersome but more intuitive for analysis based on Hoeffding's inequality.
In addition, the characterization itself does not have a closed form that identifies the number of phase-transition points.
In Sect.~\ref{sect:proofofmaintheorem}, we investigate the characterization obtained in Sect.~\ref{sect:lemproof} and show how it implies the clean characterization in our main result in Theorem~\ref{thm:main}.
The corner case where $p$ is exactly at one of the phase-transition points is analyzed in a completely different way.
In Sect.~\ref{sect:polyhedron}, we discuss the polyhedron approach proposed by~\citet{xia2021likely} and show how it can be applied to analyze the corner case when $p$ is at a phase-transition point.
Finally, we conclude our paper in Sect.~\ref{sect:conclusion}.

\section{An Intermediate Characterization}
\label{sect:lemproof}
%To prove Theorem~\ref{thm:main}, we first provide an intermediate characterization (Lemma~\ref{lem:main}) on the likelihood of AJR.
%The characterization is a more direct observation.
%However, it is much more cumbersome than Theorem~\ref{thm:main}, and it does not provide a clear picture of the number of regimes or phase-transition points.

%In Sect.~\ref{sect:proofofmaintheorem}, we will show that the characterization in Lemma~\ref{lem:main} implies our main theorem, namely Theorem~\ref{thm:main}. 
%In this section, we discuss this intermediate characterization.
%We will begin this section by discussing some observations based on Hoeffiding's inequality.
%These observations will naturally lead to our intermediate characterization in Lemma~\ref{lem:main}, and we will formally present the lemma after the discussions.
In this section, we provide an intermediate characterization (Lemma~\ref{lem:main}) on the likelihood of AJR.
This characterization will eventually lead to Theorem~\ref{thm:main}, which is discussed in Sect.~\ref{sect:proofofmaintheorem}.

To introduce the intermediate characterization, we begin by describing some very high-level ideas.
Firstly, the likelihood of the existence of an AJR committee is $1-o(1)$ for $p$ being sufficiently small or sufficiently large.
For small $p$, with high probability, there is no cohesive group at all, so every winning committee $W$ automatically satisfies AJR.
For large $p$, every winning committee $W$ is likely to be approved by many voters, making the average utility for every cohesive group large, in which case AJR is also satisfied.

\paragraph{Identifying the threshold for the existence of cohesive groups.}
It is easy to find the threshold for $p$ below which no cohesive group exists.
Consider a set $L$ of $\ell$ candidates.
Each voter approves all the candidates in $L$ with probability $p^\ell$.
Thus, the expected number of voters approving all candidates in $L$ is $np^\ell$.
When $p < \sqrt[\ell]{\frac{\ell}{k}}$, we have $np^\ell< \ell\cdot \frac{n}{k}$. 
By Hoeffding's inequality, the probability that there exists a group of $\ell \cdot \frac{n}{k}$ voters that approve every candidate in $L$ is $\Theta(\exp(-n))$.
By applying a union bound over all those $\binom{m}{\ell}$ sets of $\ell$ candidates and noticing that $\binom{m}{\ell}$ is a constant, the probability an $\ell$-cohesive group exists is $o(1)$.
On the other hand, when $p > \sqrt[\ell]{\frac{\ell}{k}}$, we have $np^\ell>\ell\cdot\frac nk$.
By Hoeffding's inequality, the number of voters approving all candidates in $L$ is at least $\ell\cdot\frac nk$ with probability $1-\Theta(\exp(-n))$, in which case $\ell$-cohesive groups exist with high probability.
Therefore, the value $\sqrt[\ell]{\frac{\ell}{k}}$ is the threshold for the existence of $\ell$-cohesive groups.
Notice that this value minimizes at $\ell=1$.
The threshold for the existence of cohesive groups (with all $\ell$'s) is $\frac1k$.
This is exactly where the value $p_1^\ast$ in Theorem~\ref{thm:main} comes from. 

\paragraph{Analysis for the case where $p$ is above the threshold $\sqrt[\ell]{\frac{\ell}{k}}$.}
Following our analysis in the previous paragraph, when $p \ge \sqrt[\ell]{\frac{\ell}{k}}$, for every set $L$ of $\ell$ candidates, there is at least $\Theta(1)$ probability that an $\ell$-cohesive group towards $L$ exists. 
To characterize how large the value $p$ must be to guarantee the existence of an AJR committee, let us consider the worst-case scenario where $W\cap L=\emptyset$, i.e., consider every $\ell$-cohesive group where none of its common approved candidates is selected in $W$.
Intuitively, this is the ``worst-case'' as such a cohesive group is least satisfied and most likely to violate the AJR condition.
The analysis with $W\cap L\neq\emptyset$ requires some extra analysis, and we skip it at this moment for the purpose of a more intuitive discussion.
In the following, we fix $W=\{1,\ldots,k\}$ and $L=\{k+1,\ldots,k+\ell\}$. %\qishen{Why consider this ``worst'' case but not other cases? Or we need to refer to some proposition saying that it suffice to consider the worst case.}

Among those voters who approve $L$, every $\ell\cdot\frac nk$ of them form an $\ell$-cohesive group.
To make sure AJR is satisfied, we need that those $\ell\cdot\frac nk$ voters in the group with \emph{minimum utilities} have an average utility of at least $\ell$.
To form such a group, we iteratively ``pick'' voters with minimum utilities among those voters who approve $L$.
We first pick those voters with utility $0$.
These voters approve all candidates in $L$ and no candidate in $W$.
The expected number of this type of voters is $np^\ell(1-p)^k$.
Next, we pick those voters with utility $1$.
These are the voters who approve all candidates in $L$ and exactly one candidate in $W$.
The expected number of them is $np^\ell\cdot\binom{k}{1}p(1-p)^{k-1}=\binom{k}1p^{\ell+1}(1-p)^{k-1}\cdot n$.
In general, the expected number of voters who approve all candidates in $L$ and exactly $t$ candidates in $W$ is $\binom{k}{t}\cdot p^{\ell + t}\cdot (1 - p)^{k - t}\cdot n$.
We keep this kind of greedy iterative selections until the total number of voters selected meets $\ell\cdot\frac nk$.
This process will stop at some $\topp$ where the number of voters in the $\ell$-cohesive group with utilities at most $\topp$ exceeds $\ell\cdot\frac nk$.
More precisely, $\topp$ is the smallest integer such that
\begin{equation}\label{eqn:tlandnl}
\ntopp := \sum_{t = 0}^{\topp} \binom{k}{t}\cdot p^{\ell + t}\cdot (1 - p)^{k - t}\cdot n \ge \ell \cdot \frac{n}{k}.
\end{equation}
Here, $\ntopp$ is the expected number of voters in the cohesive group with utilities at most $\topp$.
Finally, we need to remove $\ntopp-\ell\cdot\frac nk$ voters from those voters with utility $\topp$ to make the number of voters exactly $\ell\cdot\frac nk$.
This completes the selection of $\ell\cdot\frac nk$ voters in the $\ell$-cohesive group with minimum utilities.
Fig.~\ref{fig:jun6_1} illustrates the selection and the notations $\topp$ and $\ntopp$.

\begin{figure}[htbp]
    \centering
    \begin{tikzpicture}[scale=1]
    \definecolor{fancyblue}{RGB}{127,172,204}  
    \definecolor{fancyyellow}{RGB}{246,189,78}  
    \definecolor{fancyred}{RGB}{233,108,102}  
    \definecolor{fancygreen}{RGB}{31,145,158} 
    \draw[fancyblue, <->] (0,0.75) -- (4.2*1.5,0.75);
    \draw[fancyblue, <->] (0,-0.25) -- (5*1.5,-0.25);
    
    \draw[line width=0.1mm, rounded corners] (0,0) rectangle ++(6*1.5,0.5);
    \fill[fancyblue, rounded corners, opacity=0.5] (0,0) rectangle ++(5*1.5,0.5);

    \draw[black, line width=0.25mm] (6*1.5,-0.5) -- (6*1.5,1);
    \draw[black, line width=0.25mm] (-0.0*1.5,-0.5) -- (-0.0*1.5,1);

    \draw[black, line width=0.25mm] (5*1.5,-0.5) -- (5*1.5,0.5);
    \draw[black, line width=0.25mm] (4.2*1.5,0) -- (4.2*1.5,1);
    \draw[black, line width=0.25mm] (1*1.5,0) -- (1*1.5,0.5);
    \draw[black, line width=0.25mm] (2*1.5,0) -- (2*1.5,0.5);
    \draw[black, line width=0.25mm] (3*1.5,0) -- (3*1.5,0.5);
    \draw[black, line width=0.25mm] (4*1.5,0) -- (4*1.5,0.5);
    %\draw[black, line width=0.5mm] (5*1.5,0) -- (5*1.5,0.5);
    %\draw[black, line width=0.5mm] (6*1.5,0) -- (6*1.5,0.5);
    
    \node[above] at (2.1*1.5,0.75) {$\displaystyle |V| = \ell \cdot \frac{n}{k}$};
    \node[below] at (2.5*1.5,-0.25) {$\displaystyle n_\ell$};

    \node[above] at (0.5*1.5,0) {$\displaystyle 0$};
    \node[above] at (1.5*1.5,0) {$\displaystyle 1$};
    \node[above] at (2.5*1.5,0) {$\displaystyle 2$};
    \node[above] at (3.5*1.5,0) {$\displaystyle \cdots$};
    \node[above] at (4.5*1.5,0) {$\displaystyle t_\ell$};
    \node[above] at (5.5*1.5,0) {$\displaystyle t_\ell+1$};
\end{tikzpicture}
    \caption{Illustration of construction group $V$. Numbers in the chunk represent the number of approved candidates in $W$ for each chunk of voters.}
    \label{fig:jun6_1}
\end{figure}
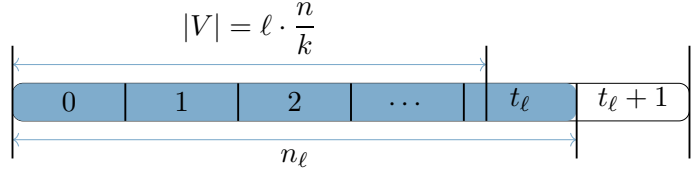

Notice that the above analysis and the voters' distribution shown in Fig.~\ref{fig:jun6_1} are all with respect to \emph{expected} numbers of voters.
However, given that the maximum value of $\ell$ (which is $k$) is a constant, we have a constant number of voter chunks in Fig.~\ref{fig:jun6_1}.
By a combination of Hoeffding's inequality and union bounds, the actual voter distribution follows the description of Fig.~\ref{fig:jun6_1} with high probability.

Lastly, we compute the expected average utility of those $\ell\cdot\frac nk$ selected voters.
As mentioned before, the expected number of voters who approve all candidates in $L$ and exactly $t$ candidates in $W$ is $\binom{k}{t}p^{\ell+t}(1-p)^{k-t}n$, which gives the expected total utility of $t\cdot\binom{k}{t}p^{\ell+t}(1-p)^{k-t}n$.
The last set of voters with utility exactly $\topp$ is ``truncated'' such that  $\ntopp-\ell\cdot\frac nk$ voters with utility $\topp$ are removed.
Putting together, the expected average utility for those $\ell\cdot\frac nk$ agents approving $L$ with minimum utilities is given by
\begin{equation*}
    \utopp = \frac{k}{\ell \cdot n}\cdot \left( \sum_{t = 0}^{\topp} t \cdot \binom{k}{t}\cdot p^{\ell + t}\cdot (1 - p)^{k - t}\cdot n - \topp\cdot (\ntopp - \ell\cdot n/k)\right). 
\end{equation*} 

If we have a large enough $p$ such that $u_\ell>\ell$ holds for all $\ell \in [k]$, then the least satisfied $\ell\cdot\frac nk$ voters who approve $L$ meets the AJR requirement with high probability.
Precisely, this probability is $1-\Theta(\exp(-n))$ by Hoeffding's inequality, as $u_\ell>\ell$ implies that the \emph{total} utility of these $\ell\frac nk$ voters is greater than $\ell\cdot\ell\frac nk$ by $\Theta(n)$.
By taking a union bound over all sets of $L$, we conclude that $W=\{1,\ldots,k\}$ is an AJR committee with high probability.

On the other hand, if $p$ is not large enough such that $u_\ell<\ell$ for some $\ell$, then, with high probability, those $\ell\cdot\frac nk$ voters approving $L$ with minimum utilities, which form an $\ell$-cohesive group, fail to meet the AJR requirements.
By taking a union bound over the choice of $W$, we conclude that an AJR committee does not exist with high probability.

Putting the above intuitions together, we have the following intermediate characterization for the likelihood of an AJR committee.

\begin{lemma}\label{lem:main}
    For any constant $m, k$, and $p$, the likelihood that there exists an AJR committee $W$  has the following trichotomy. %Consider all $\ell=1,\ldots,k$ such that $p^{\ell} \ge \frac{\ell}{k}$. 
    Let $\mathcal{L}=\{\ell\mid 1\leq\ell\leq \min\{k,m-k\}, p^{\ell}\geq\frac\ell{k}\}$.
    Let $\topp$ be the smallest positive integer such that $\ntopp = \sum_{t = 0}^{\topp} \binom{k}{t}\cdot p^{\ell + t}\cdot (1 - p)^{k - t}\cdot n \ge \ell \cdot \frac{n}{k}$. Let 
    \begin{equation}\label{eqn:lemmain}
        \utopp = \frac{k}{\ell \cdot n}\cdot \left( \sum_{t = 0}^{\topp} t \cdot \binom{k}{t}\cdot p^{\ell + t}\cdot (1 - p)^{k - t}\cdot n - \topp\cdot (\ntopp - \ell \cdot n/k)\right).
    \end{equation} 
    Then,
    \begin{enumerate}
        \item If $p < \frac{1}{k}$ (which means $\mathcal{L}=\emptyset$), $p = 1$, or $\utopp > \ell$ for every $\ell\in\mathcal{L}$, then the likelihood for the existence of an AJR committee is $1 - o(1)$;
        \item If there exists an $\ell\in\mathcal{L}$ such that $\utopp < \ell$, then the likelihood for the existence of an AJR committee is $o(1)$; 
        \item Otherwise, when $\utopp \ge \ell$ for every $\ell\in\mathcal{L}$ and there exists an $\ell\in\mathcal{L}$ such that $\utopp = \ell$, the likelihood for the existence of an AJR committee is $\Theta(1)$ and $1 - \Theta(1)$. 
    \end{enumerate}
\end{lemma}
%\qishen{$\ell$ should be no larger than $\min(k, m-k)$. Change it in the Lemma and the propositions.}

%The remaining part of this section proves this lemma formally.The proof proceeds in three steps. Proposition~\ref{prop: u>1} in Sect.~\ref{sect:prop1} characterizes the $\utopp > \ell$ case and shows that the likelihood of AJR committee's existence is $1 - o(1)$. \qishen{We should mention that the trivial cases have been proved somewhere.} Proposition~\ref{prop: u < 1} in Sect.~\ref{sect:prop2} characterizes the $\utopp < \ell$ case and shows that the likelihood of AJR committee's existence is $o(1)$.

The proof for the first two parts of Lemma~\ref{lem:main} follows the same high-level ideas described before Lemma~\ref{lem:main}.
We defer the proof to Appendix~\ref{append:lemproof}.
The third part with $\utopp=\ell$ is analyzed by applying the techniques of \emph{polyhedron approach} in~\citet{xia2021likely}, and it is handled in Sect.~\ref{sect:polyhedron}.

In Sect.~\ref{sect:proofofmaintheorem}, we further analyze the relationship between $\utopp$ and $\ell$ and show how to obtain the much cleaner characterization in Theorem~\ref{thm:main}.

\paragraph{Only $1$-Cohesive Groups Matter.}
During the proof of Thoerem~\ref{thm:main} in the coming section, we discover an interesting phenomenon stated in Proposition~\ref{prop:ell>=2}: for every $\ell\geq 2$, if $p$ is large enough such that $\ell$-cohesive groups exist, i.e., $p^\ell\geq\frac\ell{k}$, then we automatically have $u_\ell>\ell$.
This means that the pre-conditions for (2) and (3) in Lemma~\ref{lem:main} can only hold for $\ell=1$. 
For $\ell\geq 2$, the value of $p$ that is large enough for the existence of an $\ell$-cohesive group is also large enough for this $\ell$-cohesive group to satisfy AJR for any given winning committee of $k$ candidates.
Therefore, AJR can only fail for those $1$-cohesive groups under the Erd\H{o}s--R\'enyi bipartite model.
We are not sure if the phenomenon that $1$-cohesive groups are the hardest to satisfy holds under other random models, and it is an interesting future direction to see how broadly this observation can be applied.

\section{Proof of Theorem~\ref{thm:main}}
\label{sect:proofofmaintheorem}
In this section, we take a closer look to the characterization in Lemma~\ref{lem:main} and reveal some relationships between $p$, $\utopp$, and $\ell$, which enable us to conclude Theorem~\ref{thm:main}.

One challenge for obtaining Theorem~\ref{thm:main} from Lemma~\ref{lem:main} is that the definitions of $t_\ell$ and $n_\ell$, which are given in Equation (\ref{eqn:tlandnl}), are complicated and difficult to analyze.
This makes the characterization of $\utopp$ (which depends on $t_\ell$ and $n_\ell$) difficult.
To overcome this, we begin by viewing $\utopp$ in a slightly different way: instead of letting the summation to up to the particular value $\topp$ defined in Lemma~\ref{lem:main}, we consider a general ``stopping value'' $T$ instead of $t_\ell$.
That is, we define
\begin{equation}\begin{aligned}
    U(T) &= \frac{k}{\ell}\left( \sum_{t=0}^{T} t\binom{k}{t} p^{\ell+t}(1-p)^{k-t}-T\left( \sum_{t=0}^{T} \binom{k}{t} p^{\ell+t}(1-p)^{k-t} -\frac{\ell}{k}\right) \right)\\
    &= T-\frac{k}{\ell}\left( \sum_{t=0}^{T} (T-t)\binom{k}{t} p^{\ell+t}(1-p)^{k-t} \right).
\end{aligned}\end{equation}
Notice that, compared with Equation (\ref{eqn:lemmain}) in Lemma~\ref{lem:main}, we have substituted $\ntopp$ to the equation, replaced $\topp$ by a general input value $T$, and canceled $n$.

Recall that $\topp$ is the minimum utility value such that the expected number of voters who approve $L$ with utilities at most $\topp$ exceeds $\ell\cdot\frac nk$.
When $T<\topp$, we have not selected enough number of voters to reach $\ell\cdot\frac nk$.
In this case, $\sum_{t=0}^{T} \binom{k}{t} p^{\ell+t}(1-p)^{k-t} -\frac{\ell}{k}<0$, and $U(T)$ treats those $(\ell\cdot\frac nk-\sum_{t=0}^{T} \binom{k}{t} p^{\ell+t}(1-p)^{k-t}\cdot n)$ extra voters as they receive utility $T$.
Since in reality their utilities are between $T$ and $\topp$, $U(T)$ has underestimated $\utopp$.
On the other hand, when $T>\topp$, we have selected more voters than needed.
In this case, $\sum_{t=0}^{T} \binom{k}{t} p^{\ell+t}(1-p)^{k-t} -\frac{\ell}{k}>0$, and $U(T)$ excludes these $(\sum_{t=0}^{T} \binom{k}{t} p^{\ell+t}(1-p)^{k-t} -\frac{\ell}{k})n$ voters as if they receive utility $T$.
Since in reality these agents' utilities is between $\topp$ and $T$, $U(T)$ again underestimates $\utopp$.
Therefore, we have the following proposition. The full proof is in Appendix~\ref{append:propU}. 

\begin{proposition}\label{prop:U}
    For any $k, \ell$, and $p$, we have $\displaystyle u_\ell = \max_{1\le T\le k} U(T)$.
\end{proposition}
% \begin{proof}
% The proof follows the ideas described above, and is deferred to Appendix~\ref{append:propU}.
% \end{proof}

With this new interpretation of $\utopp$, we are ready to prove Theorem~\ref{thm:main}.
Our proof consists of two parts.

In the first part, we will show the following proposition.
\begin{proposition}\label{prop:ell>=2}
    For each $\ell\geq 2$, if $p^\ell\geq\frac\ell{k}$, then $\utopp>\ell$.
\end{proposition}
This proposition implies that, for each $\ell\geq 2$, either with high probability no $\ell$-cohesive group exists (in the case $p^\ell<\frac\ell{k}$) or with high probability the AJR condition is met for every $\ell$-cohesive group ($\utopp>\ell$ as suggested by the proposition).
As a result, we do not need to worry about $\ell$-cohesive groups for $\ell\geq2$ at all.
This proposition also implies, in Lemma~\ref{lem:main}, the pre-condition for (2) and (3) can only happen with $\ell=1$.

In the second part, we study the case with $\ell=1$ and prove the following proposition.
\begin{proposition}\label{prop:ell=1}
    When $\ell=1$ and $1>p>\frac{1}{k}$,
    \begin{enumerate}
        \item $U(2)<1$ implies $u_1<1$;
        \item $U(2)=1$ implies $u_1=1$;
        \item $U(2)>1$ implies $u_1>1$.
    \end{enumerate}
    In addition, the equation $U(2)=1$ has exactly one root $p^\ast$ in $(\frac1k,1)$, and
    \begin{enumerate}
        \item $U(2)<1$ when $\frac1k<p<p^\ast$, and
        \item $U(2)>1$ when $p>p^\ast$.
    \end{enumerate}
    When $\ell=1$, $p=\frac1k$, we have $1=U(k)>U(k-1)>\cdots>U(1)$.
\end{proposition}

First of all, notice that we always have $\ell=1\in\mathcal{L}$ (see Lemma~\ref{lem:main} for definition of $\mathcal{L}$) when $1>p\geq\frac1k$.
Therefore, in the proposition above, $u_1>1$, $u_1<1$, and $u_1=1$ precisely correspond to (1), (2), and (3) in Lemma~\ref{lem:main}, respectively.

The first part of the proposition implies that the second function $U(2)$ in the family $\{U(T)\}_{T=1,\ldots,k}$ characterizes the relationship between $\utopp$ and $\ell$ (which is just $u_1$ and $1$ since Proposition~\ref{prop:ell>=2} tells us we do not need to look at $\ell\geq2$).
There is a complicated reason why $U(2)$ (instead of $U(T)$ for other values of $T$) is so special here. We will provide intuitions and formal proofs later.
%For a very high-level reason why $U(2)$ (instead of $U(T)$ for other values of $T$) is so special here, we can find out that, for those values of $p$ such that $u_1<1$, we have either 1) $U(2)$ is the maximum among $U(T)$, or 2) $U(2)<1$ ($U(2)>1$ resp.) implies $U(T)<1$ ($U(T)>1$ resp.) for other values of $T$.
%We plot the curves of $U(T)$ for different values of $T$ in Fig.~\ref{fig:U}.
%It can be observed that, whenever $U(2)$ is below the horizontal line $U(T)=1$, the curve of $U(2)$ is always above the curve for any other $U(T)$.
The second part of the proposition gives a precise characterization of $U(2)$.
By noticing that Equation~(\ref{eqn:thmmain}) in Theorem~\ref{thm:main} is exactly the equation $U(2)=1$, we conclude Theorem~\ref{thm:main}, except for the only corner case $p=\frac1k$.
The last part discusses the corner case $p=\frac1k$.
By Proposition~\ref{prop:U} and Proposition~\ref{prop:ell>=2}, it implies $\utopp\geq\ell$ always holds and $\utopp=\ell$ for $\ell=1$.
Lemma~\ref{lem:main} says that the likelihood of an AJR committee is $\Theta(1)$ and $1-\Theta(1)$, which agrees with the description of Theorem~\ref{thm:main}.

Now, it remains to prove the two propositions.
Before this, we first argue that we can assume $\ell<k$ without loss of generality.
For $\ell=k$, according to Lemma~\ref{lem:main}, the only case we need to take such an $\ell$ into consideration is when $p=1$ (as this is the only $p$ making $p^\ell\geq\frac\ell k$).
In this case, it is obvious that AJR committees exist with probability $1$, which also agrees with our characterization in Theorem~\ref{thm:main}, as Proposition~\ref{prop:ell=1} implies that $p_2^\ast$ in Theorem~\ref{thm:main} is strictly less than $1$.
In the remaining part of this section, we will assume $\ell < k$.
%\qishen{We can add a high-level idea of the proof: Bound the open form by closed form. Use derivatives and monotonicity to eliminate parameters. No this may not be that important anyway.}

\subsection{Proof of Proposition~\ref{prop:ell>=2}}

For $\ell\ge 2$, we will show that all $\ell$-cohesive groups always satisfy the AJR condition.

Since we have assumed $k>\ell$, we have $\ell+1\in\{1,\ldots,k\}$ and $u_\ell = \max\limits_{1\le T\le k} U(T) \ge U(\ell+1)$.
It then suffices to show that $U(\ell+1)>\ell$ for all $\ell\geq 2$, that is,
    $$
        U(\ell+1) = \ell+1-\frac{k}{\ell}\left( \sum_{t=0}^{\ell} (\ell+1-t)\binom{k}{t} p^{\ell+t}(1-p)^{k-t} \right) >\ell.
    $$

Rearranging the inequality, we obtain
    \begin{equation}\label{eqn:obj-ell>=2}
        \frac{kp^\ell}{\ell}\left( \sum_{t=0}^{\ell} (\ell+1-t)\binom{k}{t} p^{t}(1-p)^{k-t} \right)<1.
    \end{equation}

We perform the following calculations on the left-hand side of (\ref{eqn:obj-ell>=2}). 
\begin{align*}
    &\frac{kp^\ell}{\ell}\left( \sum_{t=0}^{\ell} (\ell+1-t)\binom{k}{t} p^{t}(1-p)^{k-t} \right)\\
    =& \frac{kp^\ell}{\ell}\left( \sum_{t=0}^{\ell} \frac{k\cdots(k-t+1)}{\ell\cdots(\ell-t+2)}\cdot(\ell+1-t)\frac{\ell\cdots(\ell-t+2)}{t!} p^{t}(1-p)^{k-t} \right)\\
    \leq& \frac{kp^\ell}{\ell}\left( \sum_{t=0}^{\ell} \binom{k}{\ell}\cdot\binom{\ell}{t} p^{t}(1-p)^{k-t} \right)\tag{since $\binom{k}\ell\geq\frac{k\cdots(k-t+1)}{\ell\cdots(\ell-t+2)}$ (see explanations below) and $\binom{\ell}{t}=(\ell+1-t)\frac{\ell\cdots(\ell-t+2)}{t!}$}\\
    =& \frac{kp^\ell}{\ell}\binom{k}{\ell}\cdot(1-p)^{k-\ell}\tag{$\sum_{t=0}^{\ell} \binom{\ell}{t} p^{t}(1-p)^{\ell-t}=(p+(1-p))^{\ell}=1$}
\end{align*}
For the middle inequality,  $\binom{k}\ell\geq\frac{k\cdots(k-t+1)}{\ell\cdots(\ell-t+2)}$ holds because  $\binom{k}\ell/\frac{k\cdots(k-t+1)}{\ell\cdots(\ell-t+2)}=\frac{1}{\ell-t+1}\binom{k-t}{\ell-t}$ is decreasing in $k$ and substituting the minimum value of $k$, which is $k=\ell+1$, gives the ratio $1$.

Let $f(p)=\frac{kp^\ell}{\ell}\binom{k}{\ell}\cdot(1-p)^{k-\ell}$. We compute its derivative
$$
            f'(p)=\frac{k}{\ell}\binom{k}{\ell}p^{\ell-1}(1-p)^{k-\ell-1}(\ell-kp).
$$ 
We have $f'(p)<0$ since $p\ge \sqrt[\ell]{\frac{\ell}{k}}>\frac{\ell}{k}$. Thus, we have
$$
            f(p)\le f\left(\sqrt[\ell]{\frac{\ell}{k}}\right)=\binom{k}{\ell}\left(1-\sqrt[\ell]{\frac{\ell}{k}}\right)^{k-\ell}.
$$
Now, it remains to show the following proposition.

\begin{proposition}\label{conjec:complex_ineq}
    For any integers $\ell$ and $k$ with $2\le \ell<k$,
    $$
        \binom{k}{\ell}\left(1-\sqrt[\ell]{\frac{\ell}{k}}\right)^{k-\ell}<1.
    $$
\end{proposition}

The proof of Proposition~\ref{conjec:complex_ineq} is fully technical and involves some numerical checks by computer programs for small enough $k$ and $\ell$.
See Appendix~\ref{append:conjecture2}.

\subsection{Proof of Proposition~\ref{prop:ell=1}}
Recall that when $\ell=1$,
\begin{equation}\begin{aligned}
        U(T) = T-k\left( \sum_{t=0}^{T} (T-t)\binom{k}{t} p^{1+t}(1-p)^{k-t} \right)
\end{aligned}\end{equation}
and 
\begin{equation}\begin{aligned}
        u_1 &= \max_{1\le T\le k} U(T).
\end{aligned}\end{equation}

We will prove the three parts of the propositions in the reverse order.
The last part of the proposition can be proved by simple calculations.
The proof for the second part depends on the result in the last part, and it can be easily obtained by computing the derivative of $U(2)$ with respect to $p$.
The proof for the first part is the most technically challenging one, and it depends on the results for the second and third parts.
We will provide some intuitions and then the formal proof after we prove the third and the second parts.

We first prove the last part of the proposition.
\begin{proof}[Proof (last part of Proposition~\ref{prop:ell=1})]
    When $p=\frac{1}{k}$, we have
    \begin{align*}
        U(k)&=k-k\left(\sum_{t=0}^k(k-t)\binom{k}{t}p^{1+t}(1-p)^{k-t}\right)\\
        &=k-k\left(\sum_{t=0}^{k-1}(k-t)\binom{k}{t}p^{1+t}(1-p)^{k-t}\right)\tag{the last term is $0$}\\
        &=k-k^2p(1-p)\sum_{t=0}^{k-1}\binom{k-1}{t}p^{t}(1-p)^{k-1-t}\\
        &=k-k^2p(1-p)\left(p+(1-p)\right)^{k-1}=k-k^2p(1-p)=1,
    \end{align*} 
    and 
    \begin{align*}
        U(T+1)-U(T) &= 1-kp\left( \sum_{t=0}^T \binom{k}{t}p^t(1-p)^{k-t}\right)\tag{the last term in the summation of $U(T+1)$ is $0$}\\
        &= 1-\left( \sum_{t=0}^T \binom{k}{t}p^t(1-p)^{k-t}\right)\tag{$kp=1$}\\
        &>1-\left( \sum_{t=0}^k \binom{k}{t}p^t(1-p)^{k-t}\right) =0.
    \end{align*}
    Thus, $1=U(k)>U(k-1)>\cdots>U(2)>U(1)$ and $u_1=U(k)=1$. 
\end{proof}

We next prove the second part of the proposition. 

\begin{proof}[Proof (second part of Proposition~\ref{prop:ell=1})]
    We have
    $$U(2)=2-k(2p(1-p)^k+kp^2(1-p)^{k-1}).$$
    Direct calculations reveal
    \begin{equation}\label{eqn:U2'}
    \frac{\partial U(2)}{\partial p}=k(1-p)^{k-2}\left(k(k-1)p^2-2(1-p)^2\right),
    \end{equation}
    so, with $p<1$, $\frac{\partial U(2)}{\partial p}=0$ implies
    $$
        k(k-1)p^2-2(1-p)^2=0,
    $$
    which further implies
    $$
        p_0=\frac{1}{1+\sqrt{\frac{k(k-1)}2}},
    $$
    and, in addition, $U(2)$ is decreasing for $p\in(0,p_0)$ and is increasing for $p\in(p_0,1)$. 
    By the last part of Proposition~\ref{prop:ell=1} (which has just been proved), we have $U(2)<1$ at $p=\frac1k$.
    On the other hand, $U(2)=2>1$ at $p=1$.
    Since $\frac{\partial U(2)}{\partial p}$ changes signs only once, from negative to positive, on the interval $(0,1)$, we conclude the second part of the proposition.
\end{proof}

Now it remains to show the first part of the proposition.

First of all, notice that $U(1)<1$ holds for $\frac1k<p<1$.
Thus, we do not need to consider $U(1)$:
if there exists $T>1$ such that $U(T)\geq 1$ at a particular value of $p$, then $u_1=\max_{1\le T\le k} U(T)$ cannot be $U(1)$;
if $U(T)<1$ for all $T>1$ at a particular value of $p$,  then $u_1<1$ even when $u_1$ takes $U(1)$. %\qishen{I prefer to give an overall picture before the corner case. Otherwise it feels like guide you to the least exciting place into a guided tour.}
   
We next consider the easy case with $k=2$.
In this case, the family only contains $U(1)$ and $U(2)$, and it suffices to look at $U(2)$ alone.
Since we have seen $U(1)<1$ for $\frac1k<p<1$, the first part of Proposition~\ref{prop:ell=1} holds trivially (if $U(2)<1$, then $u_1=\max\{U(1),U(2)\}<1$ given $U(1)<1$; for the remaining two cases with $U(2)=1$ or $U(2)>1$, we have $u_1=\max\{U(1),U(2)\}=U(2)$ given $U(1)<1$).
In the remaining part of this section, we will assume $k\geq3$.

% Before proving the first part of Proposition~\ref{prop:ell=1} for $k\geq 3$, let us first characterize the function family $\{U(T)\}$.

We first characterize the function family $\{U(T)\}$.
\begin{figure}[htbp]
      \centering
      \subfigure{
           \includegraphics[scale=0.25]{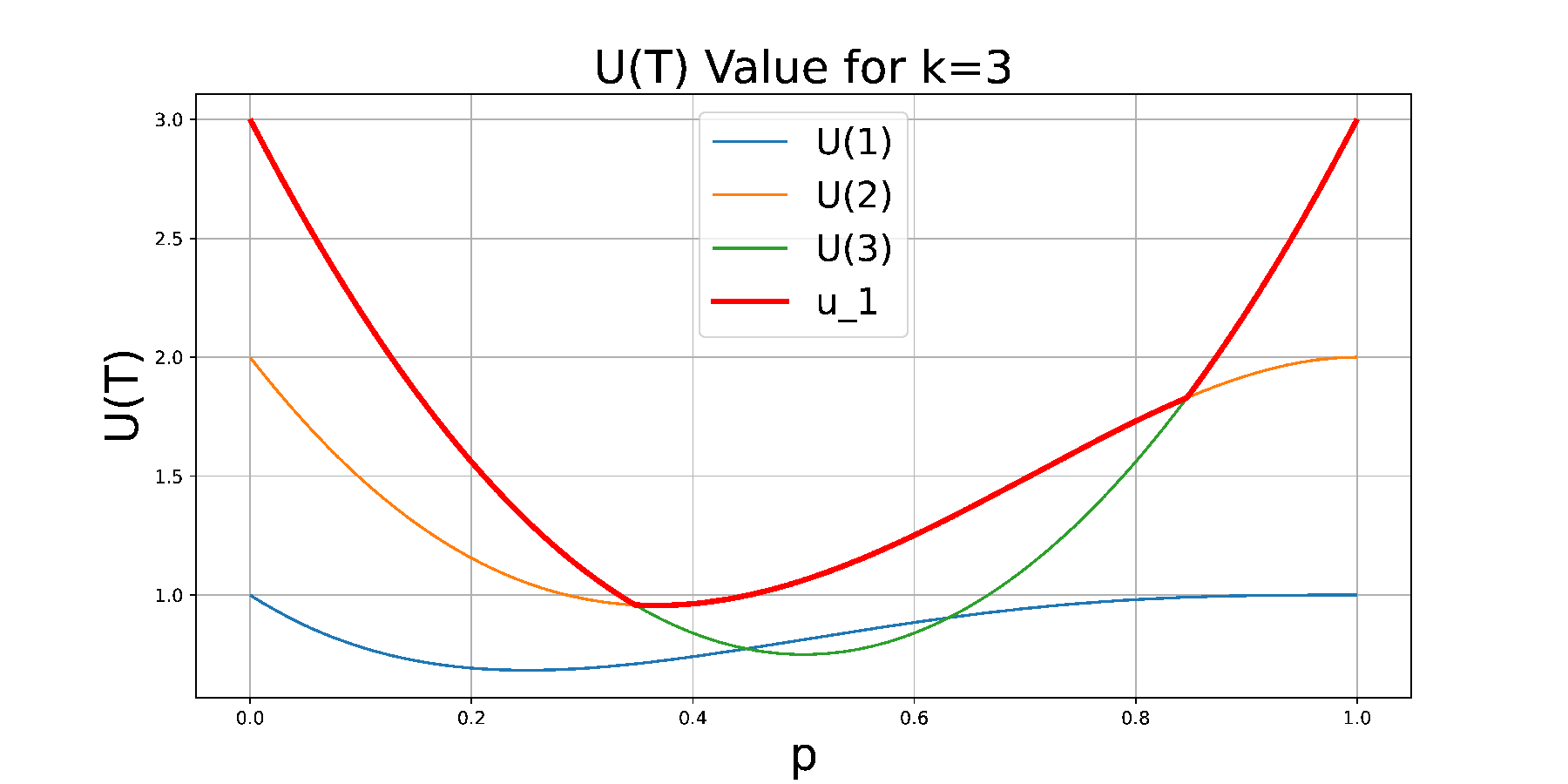}
      }
      \subfigure{
           \includegraphics[scale=0.25]{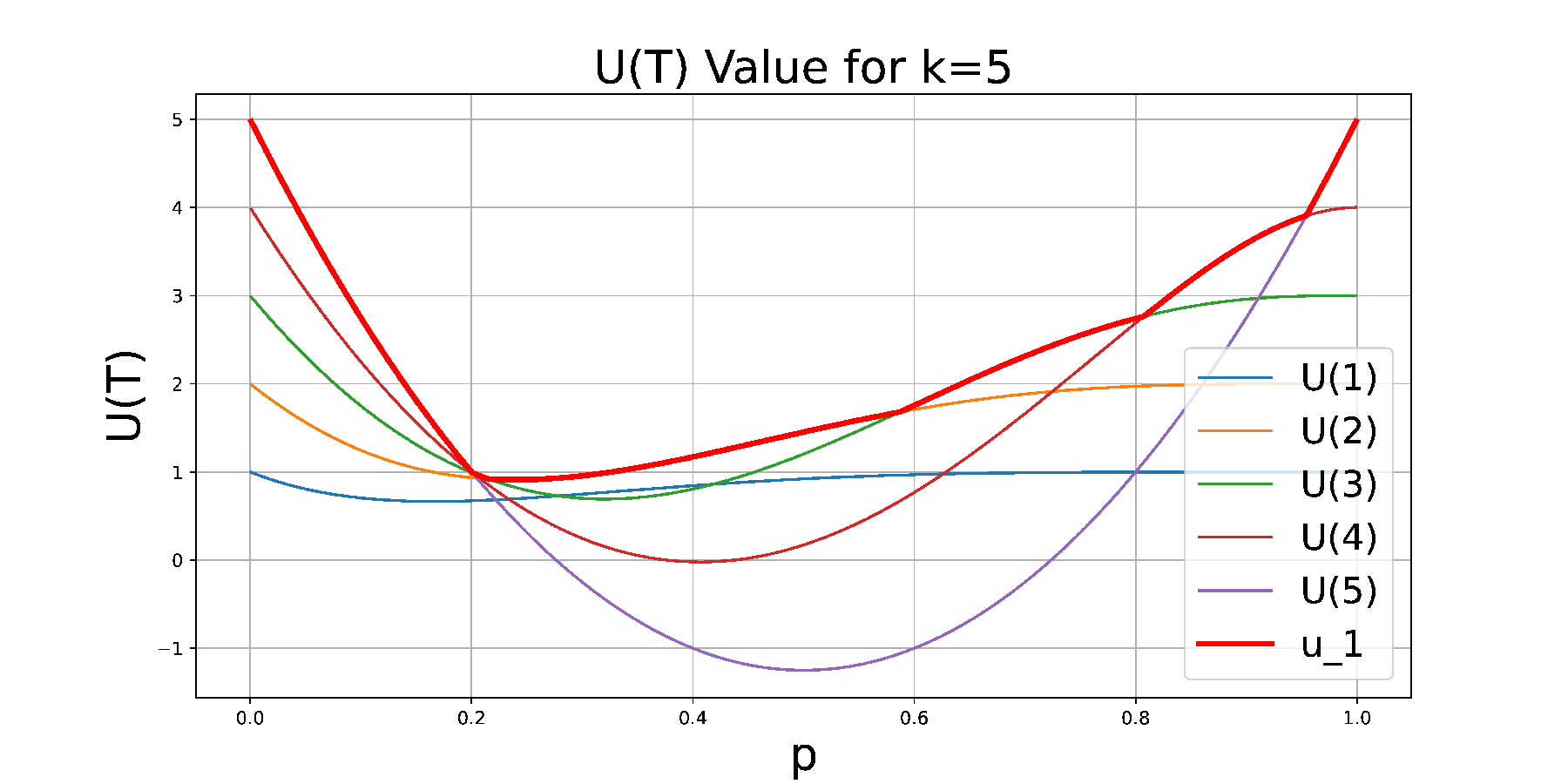}
      } 
      \subfigure{
           \includegraphics[scale=0.25]{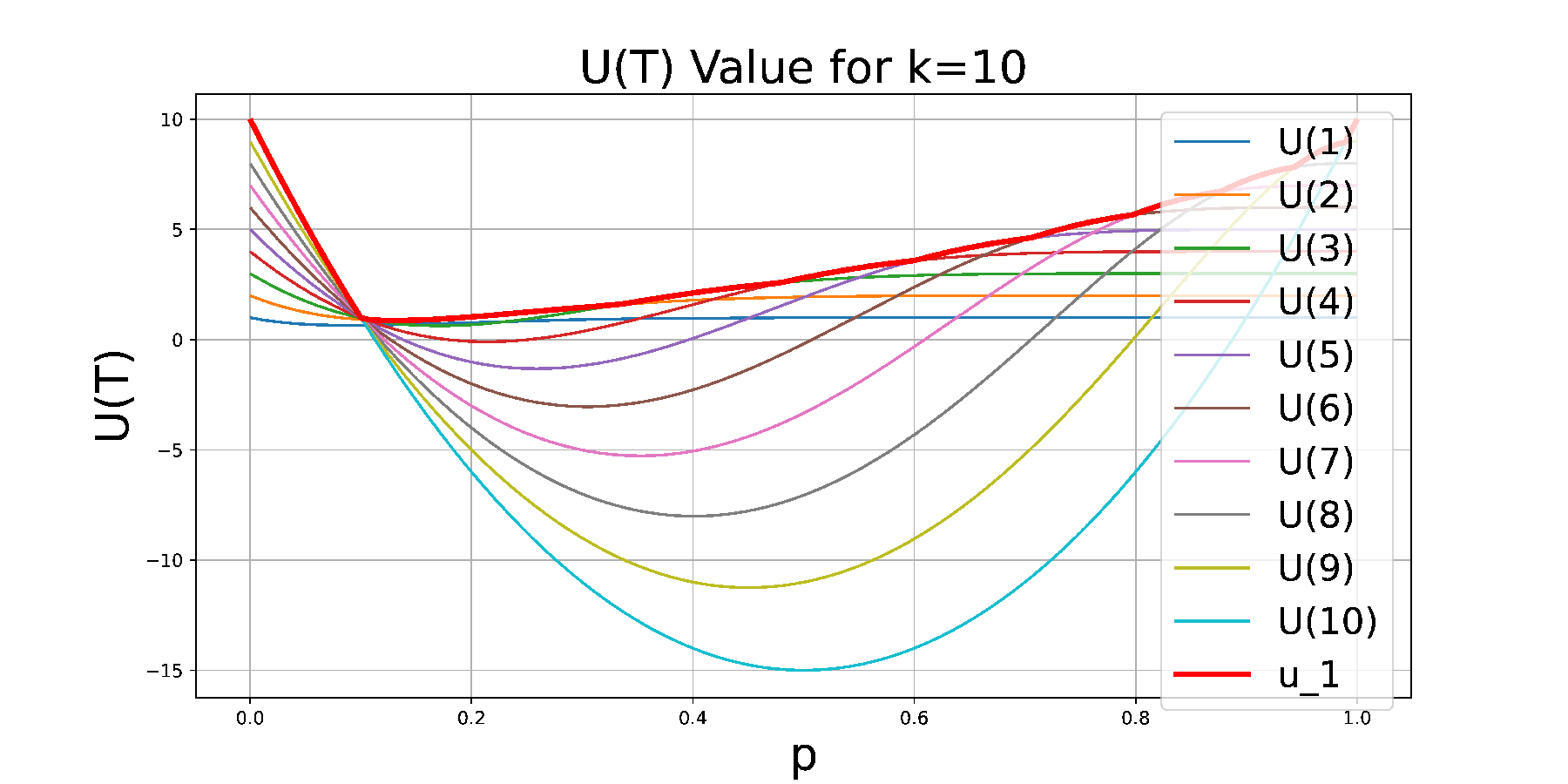}
      }
      \subfigure{
           \includegraphics[scale=0.25]{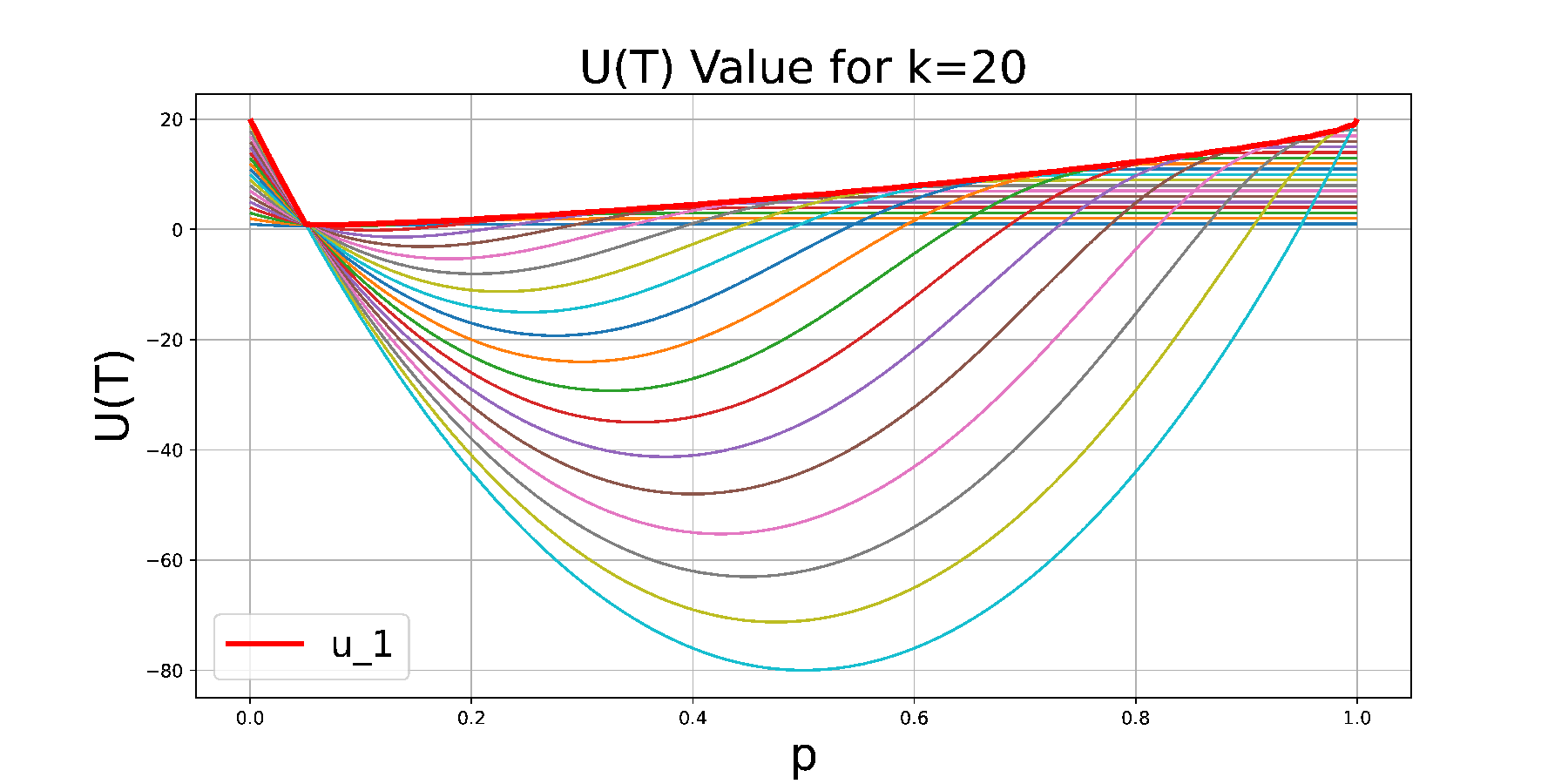}
      }
      \caption{Values of $U(T)$ with different $k$}
\label{fig:U}
\end{figure}%\qishen{probably highlight the $y = 1$ in the figure? The figures for large $k$ are very hard to follow. }
Fig.~\ref{fig:U} illustrates the curves of $U(T)$ for different values of $k$.
Recall that we want to find out for which values of $p$ we have $u_1<1$, and $u_1=\max_{1\leq T\leq k}U(T)$. 
The bold red curve in the figure represents $u_1$, which takes the maximum over all $U(T)$'s, and we can see that the curve is below $y=1$ for a small region between $p_1^\ast=\frac1k$ and some $p_2^\ast$.
To formally conclude Proposition~\ref{prop:ell=1}, we first state some observations from the pictures and show that how these observations can prove the first part of Proposition~\ref{prop:ell=1}.
We will then formally prove these observations.

Starting from $p=\frac1k$, we examine the behavior of $u_1$ from the picture as $p$ increases.
The third part of the Proposition~\ref{prop:ell=1} (which has proved already) indicates that $u_1=U(k)=1$, and when $p$ slightly increases by a little bit, all values of $U(T)$ are below $1$.
This is formally proved in Claim~\ref{prop:decreasing}, and this shows that the left-hand side threshold for $u_1<1$ is exactly at $p_1^\ast=\frac1k$.
Next, we characterize the right-hand side threshold $p_2^\ast$. We first observe from the picture that $U(2)$ reaches its minimum first after ignoring $U(1)$ (proved in Claim~\ref{prop:U(2)de_all_de}).
This implies $u_1$ is still below $1$ before $U(2)$ reaches its minimum.
In addition, $U(2)$ is the maximum among all $U(T)$'s at the time $U(2)$ reaches its minimum (Claim~\ref{prop: F(p_1)<0} and Claim~\ref{prop: U(2)_dominate_all}), and $U(2)$ continues to be the maximum until we reach a point where $U(2)=U(3)$.
We will show in Claim~\ref{prop: U(2)>1_at_2_point} that $U(2)>1$ (and so $u_1>1$) at the time we reach $U(2)=U(3)$.
Finally, when $p$ continues to increase, the value of $U(2)$ is increasing, so we have $u_1\geq U(2)>1$.
To summarize, $U(2)$ characterizes the relationship between $u_1$ and $1$:
\begin{itemize}
    \item between $\frac1k$ and the point where $U(2)$ takes its minimum, we know all $U(T)$'s are below $1$, although $U(2)$ may not be the maximum over all $U(T)$'s;
    \item between the point $U(2)$ takes its minimum and the point $U(2)=U(3)$, we will show that $U(2)$ is the maximum over all $U(T)$'s, so $u_1=U(2)$; since we know $U(2)>1$ by the time $U(2)=U(3)$, the intersection of $U(2)$ and $y=1$ happens in this regime;
    \item between the point $U(2)=U(3)$ and $p=1$, we see that $U(2)>1$, so $u_1\geq u(2)>1$ even if $U(2)$ is no longer the maximum over all $U(T)$'s.
\end{itemize}

Now we begin to prove the proposition formally.
Although these observations can be easily seen from the pictures, the proofs for many of them are rather tricky.

    \begin{claim}\label{prop:decreasing}
        When $p=\frac{1}{k}$, we have $\frac{\partial U(T)}{\partial p}<0$ for all $2\le T\le k$; that is, each $U(T)$ is decreasing at $p=\frac{1}{k}$.
    \end{claim}
    \begin{proof}
    By writing $U(T)=T-kp\left(\sum_{t=0}^T(T-t)\binom{k}{t}p^{t}(1-p)^{k-t}\right)$ and differentiating by the two parts, the term $p$ outside and the summation, we can reach the following expression by some calculations involving rewriting the summations and the combinatorial numbers. 
    The detailed calculations are available in Appendix~\ref{append:partialcalculation}.
    
        $$
           \frac{\partial U(T)}{\partial p} = -k\left( \sum_{t=0}^{T-1} \binom{k}{t} p^{t}(1-p)^{k-t-1}\left[(T-t)(1-p)-(k-t)p\right] \right)
        $$
        We will show that each term in the summation is non-positive and at least one of them is negative.
        Since each  $\binom{k}{t} p^{t}(1-p)^{k-t-1}$ is positive at $p=\frac1k$, it suffices to show that
        $(T-t)(1-p)-(k-t)p$
        is non-negative for all $t$ and positive for at least one $t$. 

        The value of $(T-t)(1-p)-(k-t)p$ is $p\left((T-t)(k-1)-k+t\right)$ since $p=\frac1k$.
        If $t=0$, we have $(T-t)(k-1)-k+t=T(k-1)-k\ge k-2> 0$.
        If $t\geq 1$, we have $(T-t)(k-1)-k+t\ge k-1-k+t\ge 0$.
\end{proof}

\begin{claim}\label{prop: U(2)_dominate_all}
    For all $1<T<k$, if $U(T)<U(T-1)$ then $U(T+1)<U(T)$.
\end{claim}
\begin{proof}
By noticing the last term in the summation of $U(T)$ is $0$, we have
    \begin{align*}
        0>U(T)-U(T-1) =&\ 1-kp\left( \sum_{t=0}^{T-1} \binom{k}{t}p^t(1-p)^{k-t}\right)\\
        >&\ 1-kp\left( \sum_{t=0}^{T} \binom{k}{t}p^t(1-p)^{k-t}\right)
        =U(T+1)-U(T),
    \end{align*}
    which implies the claim.
\end{proof}

\begin{claim}\label{prop:U(2)de_all_de}
    $\frac{\partial U(T)}{\partial p} < 0$ implies $\frac{\partial U(T+1)}{\partial p} < 0$.
\end{claim}
\begin{proof}
    Let $C_t=-k\binom{k}{t}p^{t}(1-p)^{k-t-1}<0$ for each $t=0,1,\ldots,k$.
    We have
    \begin{small}\begin{equation}\label{Partial_U(T)}
        \frac{\partial U(T)}{\partial p}=-k\sum_{t=0}^T(T-t)\binom{k}t\left((1+t)p^t(1-p)^{k-t}-(k-t)p^{1+t}(1-p)^{k-t-1}\right)=\sum_{t=0}^T(T-t)C_t[1+t-(1+k)p]<0,
    \end{equation}\end{small}
    and
    \begin{equation}\label{Partial_U(T+1)}
        \frac{\partial U(T+1)}{\partial p}-\frac{\partial U(T)}{\partial p}=\sum_{t=0}^T C_t[1+t-(1+k)p].
    \end{equation}
    We discuss the following cases:
·    \begin{itemize}
        \item $1\ge (1+k)p$. Since $C_t<0$, every term in the summation (\ref{Partial_U(T+1)}) is negative. We are done.
        \item $1+T-(1+k)p<0$, i.e., every term in the summation (\ref{Partial_U(T)}) is positive. This is a contradiction to our assumption.
        \item Otherwise, there exists $0\le t'<T$, such that $1+t'-(1+k)p\le 0$ and $1+(t'+1)-(1+k)p> 0$.
        \begin{equation}\begin{aligned}
            & &&\frac{\partial U(T+1)}{\partial p}-\frac{\partial U(T)}{\partial p}\\ 
            &= &&\sum_{t=0}^T C_t[1+t-(1+k)p]\\
            &= &&\frac{1}{T-t'} \left(\sum_{t=0}^{t'} (T-t')C_t[1+t-(1+k)p]+\sum_{t=t'+1}^T (T-t')C_t[1+t-(1+k)p]\right)\\
            &< &&\frac{1}{T-t'} \left(\sum_{t=0}^{t'} (T-t)C_t[1+t-(1+k)p]+\sum_{t=t'+1}^T (T-t)C_t[1+t-(1+k)p]\right)\\
            &= &&\frac{1}{T-t'}\frac{\partial U(T)}{\partial p}<0
        \end{aligned}\end{equation}
    \end{itemize}
    Thus, $\frac{\partial U(T+1)}{\partial p} < \frac{\partial U(T)}{\partial p}<0$.
\end{proof}

\begin{claim}\label{prop: F(p_1)<0}
    For $k\geq3$, when $\frac{\partial U(2)}{\partial p}=0$, we have $U(3)-U(2)<0$.
\end{claim}
\begin{proof}[Proof (sketch)]
    We have
    $$U(2)=2-k(2p(1-p)^k+kp^2(1-p)^{k-1})$$
    and
    $$U(3)=3-k\left(3p(1-p)^k+2kp^2(1-p)^{k-1}+\frac12k(k-1)p^3(1-p)^{k-2}\right).$$
    We have seen in the proof for the second part of Proposition~\ref{prop:ell=1} that $\frac{\partial U(2)}{\partial p}=0$ implies
    \begin{equation}\label{eqn:partial0-kp}
        k(k-1)p^2-2(1-p)^2=0,
    \end{equation}
    which further implies
    \begin{equation}\label{eqn:partial0-p}
        p=\frac{1}{1+\sqrt{\frac{k(k-1)}2}}.
    \end{equation}
    On the other hand,
    \begin{align*}
        U(3)-U(2)&=1-k^2p^2(1-p)^{k-2}\left(\frac{(1-p)^2}{kp}+1-p+\frac12(k-1)p\right)\\
        &=1-k^2p^2(1-p)^{k-2}\left(\frac{k(k-1)p^2}{2kp}+1-p+\frac12(k-1)p\right)\tag{Substituting (\ref{eqn:partial0-kp})}\\
        &=1-k^2p^2(1-p)^{k-2}(1+(k-2)p).
    \end{align*}
    Let $\phi(k,p)=k^2p^2(1-p)^{k-2}(1+(k-2)p)$. It remains to show that $\phi(k,p)>1$ for $k\geq3$ and $p$ given by (\ref{eqn:partial0-p}).

    We begin by providing a sketch to prove this.
    Notice that (\ref{eqn:partial0-p}) implies $p\approx\frac{\sqrt{2}}{k}$.
    For sufficiently large $k$, by omitting lower order terms and substituting $p\approx\frac{\sqrt{2}}{k}$, we have
    $$\phi(k,p)\approx k^2\left(\frac{\sqrt{2}}{k}\right)^2\left(1-\frac{\sqrt{2}}{k}\right)^{k}\left(1+k\cdot \frac{\sqrt{2}}{k}\right)\approx2\cdot\left(\frac1e\right)^{\sqrt2}\left(1+\sqrt2\right)\approx1.1738>1.$$
    Therefore, $\phi(k,p)>1$ holds for sufficiently large $k$.

    For a formal proof, we need to work out the constants in those approximate equalities and find out a large enough $k$ for which $\phi(k,p)$ holds.
    We see that the inequality holds for $k>1000$.
    Then the inequality for $k\leq 1000$ is checked by a computer program by enumerating all possible values of $k$.
    The details are available in Appendix~\ref{append:propF(p_1)<0}.
\end{proof}

\begin{claim}\label{prop: U(2)>1_at_2_point}
    For $k\geq 3$, when $\frac{\partial U(2)}{\partial p}>0$ and $U(3)=U(2)$, we have $U(2)>1$.
\end{claim}
\begin{proof}
    By (\ref{eqn:U2'}), $\frac{\partial U(2)}{\partial p}>0$ implies
    \begin{equation}\label{eqn:partial>0-kp}
        k(k-1)p^2 > 2(1-p)^2.
    \end{equation}
    On the other hand, simplifying $U(2)=U(3)$ gives
    \begin{equation}\label{eqn:U2=U3}
        \frac12k^2(k-1)p^3(1-p)^{k-2}=1-kp(1-p)^k-k^2p^2(1-p)^{k-1}.
    \end{equation}
    Finally,
    \begin{align*}
        U(2)&=2-k(2p(1-p)^k+kp^2(1-p)^{k-1})\\
        &=\left(1-kp(1-p)^k\right)+\left(1-kp(1-p)^k-k^2p^2(1-p)^{k-1}\right)\\
        &=\left(1-kp(1-p)^k\right)+\frac12k^2(k-1)p^3(1-p)^{k-2}\tag{by (\ref{eqn:U2=U3})}\\
        &=1+kp(1-p)^{k-2}\left(\frac12k(k-1)p^2-(1-p)^2\right)\\
        &>1,\tag{by (\ref{eqn:partial>0-kp})}
    \end{align*}
    which concludes the claim.
\end{proof}

Now we are ready to prove the first part Proposition~\ref{prop:ell=1}.
\begin{proof}[Proof (first part of Proposition~\ref{prop:ell=1})]
    To prove the first part of the proposition, we will show that there exist $p_1$ and $p_2$ with $\frac1k<p_1<p_2<1$ such that
    \begin{enumerate}
        \item for $p\in(\frac1k,p_1)$, $U(T)<1$ for all $T=2,\ldots,k$,
        \item for $p\in[p_1,p_2)$, $u_1=U(2)$, and
        \item for $p\in[p_2,1]$, $U(2)>1$.
    \end{enumerate}
    This will be sufficient for proving the first part of Proposition~\ref{prop:ell=1} due to Proposition~\ref{prop:U}.

    We will set $p_1$ to the value of $p$ where $\frac{\partial U(2)}{\partial p}=0$, which has been given in (\ref{eqn:partial0-p}).
    Notice that $p_1=\frac{1}{1+\sqrt{\frac{k(k-1)}2}}>\frac1k$.
    We will set $p_2$ to be the value of $p$ where $U(2)=U(3)$, and we will see later that $p_2>p_1$.

    For (1), by Claim~\ref{prop:decreasing} and the last part of the proposition, we have $U(T)<1$ for $T=2,\ldots,k$ in $(\frac1k,\frac1k+\varepsilon)$ for sufficiently small $\varepsilon>0$.
    On the other hand, we have seen in the proof of the second part of Proposition~\ref{prop:ell=1} that $\frac{\partial U(2)}{\partial p}<0$ for $p\in(\frac1k,p_1)$, and Claim~\ref{prop:U(2)de_all_de} implies $U(T)$ is decreasing on $(\frac1k,p_1)$ for every $T$.
    Putting together, we have $U(T)<1$ for every $T$ and $p\in(\frac1k,p_1)$.

    For (2), due to Claim~\ref{prop: U(2)_dominate_all}, it suffices to show that $U(2)>U(3)$ for $p\in[p_1,p_2)$.
     Let 
     $$F(p)=U(3)-U(2)=1-k\left(p(1-p)^k+kp^2(1-p)^{k-1}+\frac12k(k-1)p^3(1-p)^{k-2}\right)$$ 
     and 
     \begin{align*}
         F'(p)&=\frac{\partial (U(3)-U(2))}{\partial p}\\
         &=-k(1-p)^{k-3}\left[(1-p)^3+kp(1-p)^2+\frac{1}{2}k(k-1)p^2(1-p)-\frac{1}{2}k(k-1)(k-2)p^3\right].
     \end{align*}
    Let 
    \begin{align*}
        f(p)&=(1-p)^3+kp(1-p)^2+\frac{1}{2}k(k-1)p^2(1-p)-\frac{1}{2}k(k-1)(k-2)p^3.
    \end{align*}
    One can easily check that, when $k=3$, the equation $f(p)=0$ has only one solution. For general $k\ge 4$, the function $f(p)$ is a cubic function whose derivative changes signs at most twice. 
    Since $f(0)>0,f(1)<0$, and $f'(0)=k-3>0$, it is clear that $f(p)=0$ has only one solution. 
    Thus, $F(p)$ is monotonically decreasing and then monotonically increasing on $(0,1)$. Since $F(0)=F(1)=1>0$ and $F(p_1)<0$ by Claim~\ref{prop: F(p_1)<0}, there exists exactly one point $p_2\in(p_1,1)$ such that $F(p_2)=0$. Therefore, when $p\in [p_1,p_2)$, $F(p)<0$, that is, $U(3)<U(2)$.

    To see (3), Claim~\ref{prop: U(2)>1_at_2_point} tells that $U(2)>1$ when $p=p_2$.
    Since $U(2)$ is increasing on $(p_1,1)$ and $p_2>p_1$, we conclude $U(2)>1$ for $p\in[p_2,1]$.
\end{proof}

\subsection{Proof of Proposition~\ref{prop:p2}}
\label{sect:p2}
With the above analysis, we can now easily prove Proposition~\ref{prop:p2} in Sect.~\ref{sect:ourresults}.

Let $\phi(x)=k(2x(1-x)^k+kx^2(1-x)^{k-1})$ be the left-hand side of Equation~(\ref{eqn:thmmain}).
For $\ell=1$, $\phi(p)$ is just $2-U(2)$.
In the proof of the second part of Proposition~\ref{prop:ell=1}, we have seen that $U(2)$ is increasing and then decreasing on the interval $(0,1)$.
In addition, the third part of Proposition~\ref{prop:ell=1} implies $U(2)<1$ at $p=\frac1k$.
Therefore, $\phi(x)$ is increasing and then decreasing on $(0,1)$, and $\phi(\frac1k)>1$.
Noticing that $\phi(1)=0$, there is a unique solution in $(\frac1k,1)$ for the equation $\phi(x)=1$, and this solution is exactly $p_2^\ast$. In addition, any $x$ with $\phi(x)<1$ gives an upper bound to $p_2^\ast$.
Taking $x=\frac5k$ as it is in the proposition, we have
\begin{align*}
    \phi\left(\frac5k\right)&=10\left(1-\frac5k\right)^k+25\left(1-\frac5k\right)^{k-1}\\
    &<10e^{-5}+25e^{-5\frac{k-1}{k}}\tag{by the inequality $1+x<e^x$ for $x\neq0$}\\
    &\leq 10e^{-5}+25e^{-5\times\frac{2}{3}}\tag{since $k\geq3$}\\
    &<1.\tag{$10e^{-5}+25e^{-5\times\frac{2}{3}}$ is approximately $0.959$}
\end{align*}

\section{Proof of Lemma~\ref{lem:main} at Phase-Transition Points}
\label{sect:polyhedron}
The third part of Lemma~\ref{lem:main} with $\utopp=\ell$ is analyzed by applying the techniques of \emph{polyhedron approach} in~\citet{xia2021likely}.
The techniques are reviewed in Sect.~\ref{sect:technique}.
Finally, Propositions~\ref{prop: u=1neg} and~\ref{prop:u=1pos} characterize the $\utopp = \ell$ case and show that the likelihood of AJR committee's existence and non-existence are both $\Theta(1)$ respectively.

\subsection{Techniques of Polyhedron Approach}
\label{sect:technique}

\begin{definition}[Poisson multivariate variables (PMVs)~\cite{xia2021likely}]
    Given any $q,n\in\mathbb{Z}^+$ and any vector $\vec \pi=(\pi_1,\ldots,\pi_n)$ of $n$ distributions over $[q]$, an $(n,q)$-PMV, denoted by $\vec{X}_{\vec \pi}$, is the histogram of $n$ independent random variables $Y_1,Y_2,\ldots,Y_n$, where $Y_i$ follows distribution $\pi_i$.  
\end{definition}

A {\em polyhedron} $\mathcal{H}\subseteq\mathbb{R}^q$ is characterized by a matrix $A$ and a vector $\vec b$, i.e., $\mathcal{H}:=\{\vec x\in\mathbb{R}^q:A\vec x\leq \vec b\}$. Let $\dim(\mathcal{H})$ be the dimension of the polyhedron $\mathcal{H}$. We say that $\mathcal{H}\subseteq \mathbb{R}^q$ is full-dimensional if $\dim(\mathcal{H})=q$. Let $\mathcal{H}_{\le 0} = \{\vec x\in\mathbb{R}^q:A\vec x\leq \vec 0\}$ be the characteristic cone of $\mathcal{H}$. Let $\mathcal{H}_n^\mathbb{Z}$ be the set of all integral points whose $L_1$ norm is $n$ in the polyhedron $\mathcal{H}$. 

Given a set $\Pi$ of distributions over $[q]$, we can view each $\pi\in\Pi$ as a point in $\{(x_1,\ldots,x_q)\in\mathbb{R}^q\mid \sum_{i=1}^qx_i=1\}$, so $\Pi$ is a subset of $\mathbb{R}^q$. 
We say $\Pi$ is closed if $\Pi$ is a closed set in $\mathbb{R}^q$.
We further use $\CH(\Pi)$ to denote the convex hull of $\Pi$.
We say a distribution $\pi$ is strictly positive if there exists a constant $\varepsilon > 0$ such that the probability of any $i \in [q]$ under $\pi$ is at least $\varepsilon$. A distribution set $\Pi$ is strictly positive if any distribution $\pi \in \Pi$ is strictly positive.

Given $q\in\mathbb{Z}^+$, a polyhedron $\mathcal{H}\subseteq\mathbb{R}^q$, and a set $\Pi$ of distributions over $[q]$, we are interested in $\sup_{\vec\pi \in \Pi^n} \Pr[\vec{X}_{\vec \pi} \in \mathcal{H}]$ and $\inf_{\vec\pi \in \Pi^n} \Pr[\vec{X}_{\vec \pi} \in \mathcal{H}]$. 

\begin{theorem}[Smooth Likelihood of PMV-in-polyhedron~\cite{xia2021likely}]\label{thm:pmv1}
    Given any $q\in\mathbb{Z}^+$, any closed and strictly positive $\Pi$ over $[q]$, and any polyhedron $\mathcal{H}$ characterized by an integer matrix $A$, for any $n\in \mathbb{Z}^+$, 
    \begin{equation*}
       \sup_{\vec{\pi}\in \Pi^n}\mathrm{Pr}(\vec{X}_{\vec{\pi}}\in \mathcal{H}) =\begin{cases}
        0 & \mathrm{if}\ \mathcal{H}_n^\mathbb{Z}=\emptyset \\
        \exp(-\Theta(n)) & \begin{aligned}
             \mathrm{if}\ \mathcal{H}_n^\mathbb{Z}\neq\emptyset\ \mathrm{and}\ \mathcal{H}_{\leq 0}\cap \CH(\Pi)=\emptyset
        \end{aligned} \\
        \Theta\left(\sqrt{n}^{\dim(\mathcal{H}_{\leq 0})-q}\right) & \mathrm{otherwise}
    \end{cases} 
    \end{equation*}

    \begin{equation*}
        \inf_{\vec{\pi}\in \Pi^n}\mathrm{Pr}(\vec{X}_{\vec{\pi}}\in \mathcal{H}) =\begin{cases}
        0 & \mathrm{if}\ \mathcal{H}_n^\mathbb{Z}=\emptyset \\
        \exp(-\Theta(n)) & \begin{aligned}
             \mathrm{if}\ \mathcal{H}_n^\mathbb{Z}\neq\emptyset\ \mathrm{and}\ \CH(\Pi) \not\subseteq \mathcal{H}_{\leq 0}
        \end{aligned} \\
        \Theta\left(\sqrt{n}^{\dim(\mathcal{H}_{\leq 0})-q}\right) & \mathrm{otherwise}
    \end{cases}
    \end{equation*}
\end{theorem}

% The following corollary allows us to characterize the asymptotic probability when all $Y_i$ are i.i.d. on a distribution $\pi$. 
When applying the polyhedral approach to our problem, we set $q = 2^m$, where $[q]$ represents the set of all subsets of $M$, corresponding to the approval ballots of voters. Each $Y_i$ is the random variable corresponding to voter $i$'s approval ballot.  %Let $\pi$ be the Erd\H{o}s--R\'enyi bipartite model with probability $p$, and $\vec{X}_{\vec \pi}$ is the corresponding PMV. Then for all $p \in (0, 1)$, we have the following corollary for i.i.d distributions to characterize the probability that $\vec{X}_{\vec \pi}$ is in some polyhedral $\mathcal{H}$.
Let $\pi_i$ be the corresponding distribution of $Y_i$, i.e., for each $q'\in[q]=[2^m]$ that represents a subset $S_{q'}$ of candidates, $\pi_i$ assigns probability $p^{|S_{q'}|}(1-p)^{m-|S_{q'}|}$ to $q'$ under the Erd\H{o}s--R\'enyi bipartite model with parameter $p$.
We clearly have $\pi_1=\cdots=\pi_n$, and we denote this common distribution by $\pi$.
Let $\vec{X}_{\vec \pi}$ be the corresponding $(n,2^m)$-PMV where $\vec\pi=\{\pi_1,\ldots,\pi_n\}$.
We have the following corollary for i.i.d. distributions to characterize the probability that $\vec{X}_{\vec \pi}$ is in some polyhedral $\mathcal{H}$.

\begin{corollary}
\label{coro:pmv1}
    Consider any $q\in\mathbb{Z}^+$, any strictly positive distribution $\pi$ over $[q]$, and any polyhedron $\mathcal{H}$. For any  $n\in\mathbb{Z}^+$, letting $\vec{X}_{\vec \pi}$ be the $(n,q)$-PMV with $\vec\pi=(\pi_1,\ldots,\pi_n)$ and $\pi_1=\cdots=\pi_n=\pi$, we have 
    \begin{equation*}
      \mathrm{Pr}(\vec{X}_{\vec{\pi}}\in \mathcal{H}) =\begin{cases}
        0 & \mathrm{if}\ \mathcal{H}_n^\mathbb{Z}=\emptyset \\
        \exp(-\Theta(n)) & \begin{aligned}
             \mathrm{if}\ \mathcal{H}_n^\mathbb{Z}\neq\emptyset\ \mathrm{and}\ \pi \not \in \mathcal{H}_{\leq 0}
        \end{aligned} \\
        \Theta\left(\sqrt{n}^{\dim(\mathcal{H}_{\leq 0})-q}\right) & \mathrm{otherwise}
    \end{cases} 
    \end{equation*}
\end{corollary}
\begin{proof}
    We apply Theorem~\ref{thm:pmv1} with $\Pi=\{\pi\}$.
    The set $\Pi$ is finite, so it is closed.
    In addition, $\CH(\Pi)=\{\pi\}$.
    Thus, the conditions $\mathcal{H}_{\leq0}\cap\CH(\Pi)=\emptyset$ and $\CH(\Pi)\not\subseteq\mathcal{H}_{\leq0}$ for $\sup_{\vec{\pi}\in \Pi^n}\mathrm{Pr}(\vec{X}_{\vec{\pi}}\in \mathcal{H})$ and $\inf_{\vec{\pi}\in \Pi^n}\mathrm{Pr}(\vec{X}_{\vec{\pi}}\in \mathcal{H})$ respectively become the same, which is $\pi\notin\mathcal{H}_{\leq0}$.
    The corollary follows.
\end{proof}

To conclude Lemma~\ref{lem:main} for the Case $\utopp=\ell$, we will prove the following two propositions in the next two sub-sections using polyhedron approach.

\begin{proposition}
    \label{prop: u=1neg}
    For any constants $m, k, p$ and letting $\mathcal{L}$ be as defined in Lemma~\ref{lem:main}, if $\utopp \ge \ell$ for every $\ell\in\mathcal{L}$ and there exists $\ell\in\mathcal{L}$ such that $u_\ell=\ell$, then the probability that every committee $W$ fails to satisfy the AJR condition on some $\ell$-cohesive group is $\Theta(1)$. 
\end{proposition}

\begin{proposition}
    \label{prop:u=1pos}
    For any constant $m, k, p$, if $\utopp \ge \ell$ for every $1\le \ell \le k$ such that $p^\ell \ge \frac{\ell}{k}$, then the likelihood that an AJR committee exists is $\Theta(1)$. 
        
    %(\textbf{old}) For any constant $m, k, p\ge \frac{1}{k}$, if for every $\ell\le \min(k, m - k)$ with $p \ge \sqrt[\ell]{\frac{\ell}{k}}$, $\utopp \ge \ell$, then for any winning committee $W$, the likelihood that $W$ provides AJR is $\Theta(1)$. 
\end{proposition}

Both propositions share a similar proof idea. First, we construct a polyhedron $\mathcal{H}$ that characterizes a sub-event of ``an AJR committee (does not) exist'' respectively. Let $\vec{X}_{\vec{\pi}}$ be the PMV of $n$ votes following distribution $\pi$. We show that $\mathrm{Pr}(\vec{X}_{\vec{\pi}}\in \mathcal{H}) = \Theta(1)$ by applying Corollary~\ref{coro:pmv1}. Since both sub-events have a non-vanishing probability $\Theta(1)$, both ``an AJR committee exists'' and ``an AJR committee does not exist'' has probability $\Theta(1)$, and an AJR committee exists with probability $\Theta(1)$ and $1 - \Theta(1)$. 

To show $\mathrm{Pr}(\vec{X}_{\vec{\pi}}\in \mathcal{H}) = \Theta(1)$ we prove three claims. (1) $\pi \in \mathcal{H}_{\le 0}$, (2) $\mathcal{H}_n^{\mathbb{Z}} \neq \emptyset$, and (3) $\mathrm{dim}(\mathcal{H}_{\le 0}) = q$. Claim (1) and (2) ensure that  Corollary~\ref{coro:pmv1} falls into the third case, and (3) ensures that the third case implies $\Theta(1)$. 

\subsection{Proof of Proposition~\ref{prop: u=1neg}}\label{apx:u=1neg}
    We use the polyhedron approach to characterize a sub-event of ``for every committee, an underrepresented $\ell$-cohesive group exists''. For this proof, we fix an $\ell$ such that $u_\ell = \ell$. Note that $u_\ell$ and $\topp$ are determined solely on $\ell$. 

    Recall that $M$ is the set of candidates. Let $\Mwin$ be the set of all $k$-subset of $M$ and $\Mlose$ be the set of all $\ell$-subset of $M$. Let $\phi: \Mwin \to \Mlose$ be a mapping that maps every winning committee $W$ to an $\ell$-set $L_W \subseteq M\setminus W$. The existence of $\phi$ is guaranteed by $\ell \le m - k$. 
    
    For each voter $i$, let $Y_i$ be the random variable on $2^M$ denoting $i$'s approval ballot. For each $S\subseteq M$, let $x_S = |\{i\mid Y_i = S\}|$ be the number of voters whose approval ballot is exactly $S$. Note that $\sum_{S\subseteq M} x_S = n$.  Let $\pi$ be the distribution on $2^M$ in which each candidate is approved independently with probability $p$. Let $\vec X$ be the histogram of $n$ independent variables $(Y_1, Y_2, \ldots, Y_n)$, where each $Y_i$ follows $\pi$ independently. Consider the polyhedron $\mathcal{H}\subseteq\mathbb{R}^{2^m}$ defined by the following three sets of constraints. Note that $\topp$ is a constant given $m, k, \ell$, and $p$.
    Note also that each variable $x_S$ below is relaxed so that fractional values are allowed.
    \begin{align*}
    \forall W \in \Mwin, &\sum_{t = 0}^{\topp} \sum_{L_W \subseteq S, |S \cap W| = t} x_S \ge \ell /k \cdot \sum_{S\subseteq M} x_S.\\
    \forall W \in \Mwin, &\sum_{t = 0}^{\topp - 1} \sum_{L_W \subseteq S, |S \cap W| = t} x_S \le \ell /k \cdot \sum_{S\subseteq M} x_S.\\
    \forall W \in \Mwin, &\sum_{t = 1}^{\topp} \sum_{L_W \subseteq S, |S \cap W| = t} t\cdot x_S - \left(\sum_{t = 0}^{\topp} \sum_{L_W \subseteq S, |S \cap W| = t} x_S - \ell /k \cdot \sum_{S\subseteq M} x_S \right) \cdot \topp \\
    &\hspace{9.5em} \le \ell^2 /k \cdot \sum_{S\subseteq M} x_S - \topp - \frac1k.
\end{align*}
%$\mathcal{H}$ characterizes the following event: 
For each winning committee $W$, the three constraints respectively say the followings:
\begin{enumerate}
    \item there exists a group with at least $\ell\cdot n / k $ voters, all of which approve all candidates in $L_W$ as well as at most $\topp$ candidates in $W$;
    \item the number of voters that  approve all candidates in $L_W$ as well as at most $\topp - 1$ candidates in $W$ is at most $\ell\cdot n / k $; and
    \item the total utility of $\ell \cdot n / k$ voters with the minimum utility (fractional voters allowed) is strictly lower than $\ell^2 \cdot n/k - \topp$.  %In this way, for those $\lceil \ell \cdot n/k \rceil$ voters with the minimum utilities, at most $\topp$ of them are added. Hence, the total utility will be strictly lower than $\ell^2\cdot n/k$, and then the average utility will be strictly lower than $\ell$.
    After taking the ceiling function, the total utility for those $\lceil \ell \cdot n/k \rceil$ voters with the minimum utilities is strictly less than $\ell^2 \cdot n/k$ (since $\lceil \ell \cdot n/k \rceil-\ell\cdot n/k<1$ and $t_\ell$ is an upper bound on the utility of a single voter in this group).
    Notice that this implies the average utility for those $\lceil \ell \cdot n/k \rceil$ voters is strictly less than $\ell$, failing the AJR condition.
\end{enumerate}

It is not hard to verify that $\mathcal{H}$ is a sub-event of ``every winning committee $W$ fails the AJR condition on some $\ell$-cohesive group''.
Therefore, it remains to show $\Pr[\vec X \in \mathcal{H}] = \Theta(1)$.
We will show the following three claims, which imply $\Pr[\vec X \in \mathcal{H}] = \Theta(1)$ by the third case of Corollary~\ref{coro:pmv1}.

\begin{claim}
\label{claim:incone_1}
    $\pi \in \mathcal{H}_{\le 0}$. 
\end{claim}
\begin{proof}
The characteristic cone $\mathcal{H}_{\le 0}$ can be written in the following form:
\begin{align*}
    \forall W \in \Mwin, & \sum_{t = 0}^{\topp} \sum_{L_W \subseteq S, |S \cap W| = t} x_S \ge \ell /k \cdot \sum_{S\subseteq M} x_S.\\
    \forall W \in \Mwin, & \sum_{t = 0}^{\topp - 1} \sum_{L_W \subseteq S, |S \cap W| = t} x_S \le \ell /k \cdot \sum_{S\subseteq M} x_S.\\
    \forall W \in \Mwin, & \sum_{t = 1}^{\topp} \sum_{L_W \subseteq S, |S \cap W| = t} t\cdot x_S - \left(\sum_{t = 0}^{\topp} \sum_{L_W \subseteq S, |S \cap W| = t} x_S - \ell /k \cdot \sum_{S\subseteq M} x_S \right) \cdot \topp \\ & \hspace{9.5em} \le \ell^2 /k \cdot \sum_{S\subseteq M} x_S. 
\end{align*}
As $\pi$ is a distribution on $2^m$, we can represent it as a vector on $\mathbb{R}^{2^m}$, where $\pi_S$ is the probability that the approval set of an agent is exactly $S$. 

We then show that $\pi$ satisfies the three constraints above for any $W\in\Mwin$.
By definition, $\topp$ is the largest number of winners a voter in the underrepresented $\ell$-cohesive group may approve in expectation, where every $Y_i$ follows $\pi$. That is (Equation~(\ref{eqn:tlandnl})), $\topp$ is the smallest integer such that
\begin{equation*}
\ntopp := \sum_{t = 0}^{\topp} \binom{k}{t}\cdot p^{\ell + t}\cdot (1 - p)^{k - t}\cdot n \ge \ell \cdot \frac{n}{k}.
\end{equation*}
On the other hand, by the definition of $\pi$, $\sum_{L_W \subseteq S, |S \cap W| = t} \pi_S=\binom{k}{t}p^{\ell+t}(1-p)^{k-t}$, and $\sum_{S\subseteq M} \pi_S=1$
Therefore, the first two constraints hold. When $\pi$ is assigned in the third constraint, the left-hand side becomes $\ell / k \cdot \utopp$, which is $\ell^2 / k$ due to our choice of $\ell$ with $\utopp=\ell$, while the right-hand side is exactly $\ell^2 / k$. Therefore, the third constraint holds.
\end{proof}

\begin{claim}
    For all sufficiently large $n$, $\mathcal{H}$ has a inner point $\vec X_n$ such that $\|\vec X_n\|_1 = n$, and all entries in $\vec X_n$ are integers. That is, $\mathcal{H}_n^{\mathbb{Z}}\neq\emptyset$.
\end{claim}
\begin{proof}
Our construction of $\vec X_n$ consists of two steps. We first construct $\vec X_n'$, which is an inner point of $\mathcal{H}$ and has $\|\vec X_n'\|_1 = n$. Then we round $\vec X_n'$ to get the integer point $\vec X_n$. 

$\vec X_n'$ is as follows. 
For each $t = 0, 1,2, \ldots, m$ and each $S$ with $|S| = t$,
\begin{itemize}
    \item If $t= \ell$, $x_S' = p^t \cdot (1 - p) ^{m - t}\cdot n + \topp\cdot 2^m$. 
    \item If $t = 0$, $x_S' = p^t \cdot (1 - p) ^{m - t}\cdot n - \topp\cdot 2^m\cdot  \binom{m}{\ell}$. 
    \item For all other $t$, $x_S'  = p^t \cdot (1 - p) ^{m- t}\cdot n$.
\end{itemize}
Since $m$ is a constant and $\topp \le k \le m$, we can guarantee that $x_S' \ge 0$ for all sufficiently large $n$. Note that $u_{\ell} \ge \ell$ directly implies that $t_{\ell} \ge \ell$. 

$\vec X_n'$ is modified from the expectation of each $x_S$ under $\pi$. For each set $S$ of $\ell$ candidates, $\topp \cdot 2^m$  voters are switched from approving nobody to approving candidates in $S$ only. Since we do not change the total number of voters, $\|\vec X_n'\|_1 = n$ holds. 

Now we show that $\vec X_n'$ is an inner point of $\mathcal{H}$. After the modification, the left-hand side for all three constraints decreases, while the right-hand side remains unchanged. Therefore, $\vec X_n'$ strictly satisfies all constraints and is an inner point.
\begin{enumerate}
    \item For constraint 1$, \topp\cdot 2^m$ more voters are approving each $L_W$ compared to the expectation. Therefore, $\sum_{t = 0}^{\topp} \sum_{L_W\subseteq S, |S \cap W| = t} x_S' \ge \ell /k \cdot \sum_{S\subseteq M} x_S + \topp\cdot 2^m > \ell /k \cdot \sum_{S\subseteq M} x_S$.
    \item For constraint 2, by the definition of $\topp$, we have $\sum_{t = 0}^{\topp - 1} \binom{k}{t} p^{\ell + t} \cdot (1 - p)^{k - t} < \ell / k$. Therefore,
    \begin{align*}
        \sum_{t = 0}^{\topp - 1} \sum_{L_W \subseteq S, |S \cap W| = t} x_S' =&\  \sum_{t = 0}^{\topp - 1} \binom{k}{t} \cdot p^{\ell + t} \cdot (1 - p)^{k - t}\cdot n + \topp\cdot 2^m\\
        =&\ \ell /k \cdot n -\Theta(n) +\topp\cdot 2^m\\
        <&\ \ell /k \cdot \sum_{S\subseteq M} x_S.
    \end{align*}
    The first equality comes from that $x_S'$ has $\topp\cdot 2^m$ more voters approving each $L_W$ compared to the expectation. 
    The last strict inequality holds for all sufficiently large $n$. 
\item For constraint 3, for each group $L_W$, we compare the voters that approve $L_W$ between the expectation and $\vec X_n'$. There are $\topp\cdot 2^m$ more voters that do not approve any candidates in $W$ in $\vec X_n'$, while for each $1\le t\le k$, the number of voters that approve $t$ candidates in $W$ is unchanged. Therefore, when the $\ell \cdot n / k$ (fractional) voters with the minimum utility are counted, $\topp\cdot 2^m$ voters with utility 0 substitute voters with the utility strictly larger than $0$ compared to the expectation, with at least one voter with utility $\topp$. Therefore, the new total utility is at most $\ell^2\cdot n/k - \topp - \topp\cdot 2^m + 1 < \ell^2 /k \cdot \sum_{S\subseteq M} x_S' - \topp - \frac1k$.
\end{enumerate}

Now we round the fractional instance $\vec X_n'$ into an integer instance $\vec X_n$. Following the spirit of the intermediate value theorem, it is guaranteed that there exists a rounding such that $\|\vec X_n\|_1 = n$ and for every $S$, $x_S' - 1 \le x_S \le x_S' + 1$. Therefore, 
\begin{align*}
    \sum_{t = 0}^{\topp} \sum_{L_W \subseteq S, |S \cap W| = t} |x_S - x_S'| \le &\ \sum_{t = 0}^{k} \sum_{L_W \subseteq S, |S \cap W| = t} 1\\
    =&\ \sum_{L_W \subseteq S} 1\\
    =&\ 2^{m - \ell}, 
\end{align*}
 which is a constant in our setting. Then we show that $\vec X_n$ is also an inner point of $\mathcal{H}$. 
\begin{enumerate}
    \item For constraint 1, the left-hand side difference between $x_S'$ and $x_S$ is at most $2^{m - \ell}$, while $x_S'$ exceeds the threshold for at least $\topp\cdot 2^m$. 
    \begin{align*}
    \sum_{t = 0}^{\topp} \sum_{L_W \subseteq S, |S \cap W| = t} x_S \ge &\ 
        \sum_{t = 0}^{\topp} \sum_{L_W \subseteq S, |S \cap W| = t} x_S' \!-\! \sum_{t = 0}^{\topp} \sum_{L_W \subseteq S, |S \cap W| = t} |x_S' - x_S|\\ \ge&\ \ell /k \cdot \sum_{S\subseteq M} x_S + \topp\cdot 2^m - 2^{m - \ell}\\
        >&\  \ell /k \cdot \sum_{S\subseteq M} x_S.
    \end{align*}
    \item For constraint 2, the left-hand side difference between $x_S'$ and $x_S$ is at most $2^{m - \ell}$, while the gap between the left-hand size on $x_S'$ and the right-hand side is $\Theta(n)$. Therefore, for all sufficiently large $n$, (2) is satisfied, and the inequality is strict. 
    \item For constraint 3, 
    \begin{align*}
&\ \sum_{t = 1}^{\topp} \sum_{L_W \subseteq S, |S \cap W| = t} t\cdot x_S - \left(\sum_{t = 0}^{\topp} \sum_{L_W \subseteq S, |S \cap W| = t} x_S - \ell \cdot n/ k \right)\cdot \topp \\
=&\ -\sum_{t = 1}^{\topp} \sum_{L_W \subseteq S, |S \cap W| = t} (\topp - t)\cdot x_S +\topp \cdot \ell \cdot n/k\\
=&\ -\sum_{t = 1}^{\topp} \sum_{L_W \subseteq S, |S \cap W| = t} (\topp - t)\cdot \left(x_S' - (x_S' - x_S)\right) +\topp \cdot \ell \cdot n/k\\
\le &\ -\sum_{t = 1}^{\topp} \sum_{L_W \subseteq S, |S \cap W| = t} (\topp - t)\cdot x_S+\topp \cdot \ell \cdot n/k + (\topp-1)\cdot 2^{m - \ell}\\
=&\ \sum_{t = 1}^{\topp} \sum_{L_W \subseteq S, |S \cap W| = t} t\cdot x_S' \!-\! \left(\sum_{t = 0}^{\topp} \sum_{L_W \subseteq S, |S \cap W| = t} x_S' \!-\! \ell \!\cdot\! n/ k\right)\!\cdot\! \topp \!+\! (\topp-1)\cdot 2^{m - \ell}\\
\le &\ \ell^2 \cdot n/k - \topp - \topp\cdot 2^m + 1 + (\topp-1)\cdot 2^{m - \ell}\\
<&\  \ell^2 /k \cdot \sum_{S\subseteq M} x_S - \topp - \frac1k.
\end{align*}
The first inequality comes from that $t \ge 1$ so that $\topp - t \le \topp - 1$ and $\sum_{t = 0}^{\topp} \sum_{L_W\subseteq S, |S \cap W| = t} |x_S - x_S'| \le 2^{m - \ell}$.  The second inequality comes from the satisfaction of (3) of $x_S'$. 
\end{enumerate}
Therefore, we have shown that $\vec X_n$ is an inner point of $\mathcal{H}$ for all sufficiently large $n$. 
\end{proof}

\begin{claim}
    $\mathcal{H}_{\le 0}$ is full-dimensional, i.e., $\dim(\mathcal{H}_{\le 0}) = 2^m$.
\end{claim}
\begin{proof}
By using a similar analysis, we can show that $\vec X_n$ is an inner point of $\mathcal{H}_{\le 0}$. Therefore, $\mathcal{H}_{\le 0}$ is full-dimensional.
\end{proof}

With all three claims holding, we apply Corollary~\ref{coro:pmv1} and show that $\Pr[\vec X \in \mathcal{H}] = \Theta(1)$. 

% \begin{proposition}
%     \label{prop:u=1pos}
%     For any constant $m, k, p\ge \frac{1}{k}$, if for every $\ell\le \min(k, m - k)$ with $p \ge \sqrt[\ell]{\frac{\ell}{k}}$, $\utopp \ge \ell$, then for any winning committee $W$, the likelihood that $W$ provides AJR is $\Theta(1)$. 
% \end{proposition}
\subsection{Proof of Proposition~\ref{prop:u=1pos}.}\label{apx:u=1pos}
    The proof of Proposition~\ref{prop:u=1pos} resembles those of Proposition~\ref{prop: u>1} and~\ref{prop: u=1neg}. In this proof we fix a winning committee $W$. By using a similar analysis in Proposition~\ref{prop: u>1}, we can show that the probability that $W$ fails to provide AJR on any $\ell$-cohesive groups with $p < \sqrt[\ell]{\frac{\ell}{k}}$ and the probability that $W$ fails on any $L$ with $L\cap W \neq \emptyset$ is o(1). By Lemma~\ref{lem:overlap} and $\topp \ge 1$, we know that the minimum expected average utility among all $\ell$-cohesive groups towards $L$ with $L\cap W \neq \emptyset$ is strictly larger than $\ell$ ($u_{\ell, h}^* > u_{\ell - h} + h$, where $u_{\ell - h} \ge \ell - h$), which implies $o(1)$ in the likelihood of failing to provide AJR. 

    Then it remains to show that the likelihood that $W$ fails on some candidate set $L$ with $L\cap W = \emptyset$ is $1 - \Theta(1)$. We prove this through the polyhedron approach. Specifically, the polyhedron $\mathcal{H}'$ is characterized as follows. For each $\ell \in \min (k, m - k) $ with $p\ge \sqrt[\ell]{\frac{\ell}{k}}$, and each $\ell$-candidate group $L$ such that $L\cap W = \emptyset$, there are the following three constraints. In total, there will be at most $3 \sum_{\ell = 1}^{k} \binom{m}{\ell} \le 3\cdot k \cdot \binom{m}{k}$ constraints. 
    \begin{align*}
    & \sum_{t = 0}^{\topp} \sum_{L \subseteq S, |S \cap W| = t} x_S \ge \ell /k \cdot \sum_{S\subseteq M} x_S.\\
    & \sum_{t = 0}^{\topp - 1} \sum_{L \subseteq S, |S \cap W| = t} x_S \le \ell /k \cdot \sum_{S\subseteq M} x_S.\\
    & \sum_{t = 1}^{\topp} \sum_{L \subseteq S, |S \cap W| = t} t\cdot x_S - \left(\sum_{t = 0}^{\topp} \sum_{L_W \subseteq S, |S \cap W| = t} x_S - \ell /k \cdot \sum_{S\subseteq M} x_S \right) \cdot \topp \\
    & \hspace{8em} \ge \ell^2 /k \cdot \sum_{S\subseteq M} x_S
\end{align*}

$\mathcal{H}'$ characterizes the following event: for every $\ell \in\min\{k, m - k\}$ and every group $L$, 
\begin{enumerate}
    \item there exists a group with at least $\ell\cdot n / k $ voters, all of whom approve all candidates in $L$ as well as at most $\topp$ candidates in $W$;
    \item the number of voters that  approve all candidates in $L$ as well as at most $\topp - 1$ candidates in $W$ is at most $\ell\cdot n / k $; and
    \item the total utility of $\ell \cdot n / k$ voters with the minimum utility (fractional voters accepted) is at least $\ell^2 \cdot n/k$.  In this way, for the $\lceil \ell\cdot n/k \rceil$ voters with the minimum utility, the total expected utility will be at least $\ell^2 \cdot n/k + \topp\cdot (\lceil \ell\cdot n/k \rceil - \ell\cdot n/k)$. Given that $\utopp \ge \ell$, there must be $\topp \ge \ell$, since $\topp$ is the largest expected single-voter utility in the group. Therefore, the total expected utility of $\lceil \ell\cdot n/k \rceil$ voters with the minimum utility will be at least $\ell \cdot \lceil \ell\cdot n/k \rceil$, which implies the average utility being at least $\ell$. 
\end{enumerate}

Note that in the polyhedron, $t_\ell$ for different $\ell$ may be different.

Therefore, $\mathcal{H}'$ is a subevent of ``$W$ does not fail to provide AJR on any non-overlapping set of $\ell$- candidates with $W$ and $p \ge \sqrt[\ell]{\frac{\ell}{k}}$''. Therefore, showing that $\Pr[\vec{X} \in \mathcal{H}'] = \Theta(1)$ implies that the probability $W$ fails on some non-overlapping candidate set $L$ is $1 - \Theta(1)$. Then applying the union bound on all three cases where $W$ fails AJR, we have the probability that $W$ fails AJR is at most $1 - \Theta(1) + o(1) + o(1) = 1 - \Theta(1)$, which implies our conclusion.

Now we prove the claims that apply Corollary~\ref{coro:pmv1}. This resembles the proof of Proposition~\ref{prop: u=1neg}.

\begin{claim}
    $\pi \in \mathcal{H}'_{\le 0}$. 
\end{claim}
\begin{proof}
Note that $\mathcal{H}'$ does not contain any constant term. Therefore, $\mathcal{H}'_{\le 0} = \mathcal{H}'$. The proof then resembles that of Claim~\ref{claim:incone_1}.
By definition, $\topp$ is the largest number of winners a voter in the underrepresented $\ell$-cohesive group may approve in expectation where every $y_j$ follows $\pi$. Therefore, constraints 1 and 2 hold. When $\pi$ is assigned in constraint 3, the left-hand side becomes $\ell / k \cdot \utopp$, which will not exceed $\ell^2 / k$, while the right-hand side is exactly $\ell^2 / k$. Therefore, constraint 3 holds.
\end{proof}
\begin{claim}
    For all sufficiently large $n$, $\mathcal{H}'$ has a inner point $\vec X_n$ such that $\|\vec X_n\|_1 = n$, and all entries in $\vec X_n$ are integers. 
\end{claim}
\begin{proof}
Our construction of $\vec X_n$ consists of two steps. We first construct $\vec X_n'$, which is a inner point of $\mathcal{H}'$ and has $\|\vec X_n'\|_1 = n$. Then we round $\vec X_n'$ to get the integer point $\vec X_n$. 

$\vec X_n'$ is constructed by a two-step modification from the expectation. In the beginning, for every $S \subseteq M$,  $x_S'  = p^t \cdot (1 - p) ^{m - t}\cdot n$, where $t = |S|$. 
Then we do the following modification:
\begin{enumerate}
    \item For $S = \emptyset$, $x_S' \leftarrow x_S' - k \cdot (2^m - 1)$; for all other $S$, $x_S' \leftarrow x_S' + k$. For each $S \neq \emptyset$, $k$ voters approving nobody turn to approve $S$. 
    \item For all $S$ with $|S\cap W| = 0$, $x_S' \leftarrow x_S' - k^3\cdot 2^m$; for all $S$ with $|S \cap W| = 1$, $x_S' \leftarrow x_S' + k^2\cdot 2^m$. For each $S$ with $S\cap W = \emptyset$ and each candidate $j \in W$, $k^2\cdot 2^m$ voters approving exactly $S$ additionally approve $j$.
\end{enumerate}

Since both $m$ and $k$ are constants and $\topp \le k$, we can guarantee that $x_S' \ge 0$ for all sufficiently large $n$. 
We do not change the total number of voters, so $\|\vec X_n'\|_1 = n$ holds. 
Now we show that $\vec X_n'$ is an inner point of $\mathcal{H}'$ by showing that for each $L$, all three constraints are satisfied. 
\begin{enumerate}
    \item For constraint 1, modification 1 brings $k$ more voters approving at least all candidates in $L$, compared to the expectation, and modification 2 only transfers voters between different $S \supseteq L$ and does not affect the number of voters approving all candidates in $L$. Therefore, we have
    \begin{align*}
        \sum_{t = 0}^{\topp} \sum_{L \subseteq S, |S \cap W| = t} x_S' \ge \ell /k \cdot \sum_{S\subseteq M} x_S + k\cdot \sum_{t = 0}^{\topp} \sum_{L \subseteq S, |S \cap W| = t} 1 > \ell /k \cdot \sum_{S\subseteq M} x_S.
    \end{align*}
    \item For constraint 2, by the definition of $\topp$, there is $\sum_{t = 0}^{\topp - 1} \binom{k}{t} p^{\ell + t} \cdot (1 - p)^{k - t} < \ell / k$. On the other hand, modification 1 brings at most $k \cdot (2^m - 1)$ more votes approving $L$.  Therefore,
    \begin{align*}
        \sum_{t = 0}^{\topp - 1} \sum_{L \subseteq S, |S| = t + \ell} x_S' =&\  \sum_{t = 0}^{\topp - 1} \binom{k}{t} \cdot p^{\ell + t} \cdot (1 - p)^{k - t}\cdot n + k \cdot (2^m - 1)\\
        =&\ \ell /k \cdot n -\Theta(n) +k \cdot (2^m - 1)\\
        <&\ \ell /k \cdot \sum_{S\subseteq M} x_S.
    \end{align*}
    The last strict inequality holds for all sufficiently large $n$. 
\item For constraint 3, for each $L$, we compare the voters that approve $L$ between the expectation and $\vec X_n'$. Modification 1 brings $\topp$ more voters approving $L$ compared to the expectation. The total number of the increased voters will not exceed $k\cdot (2^m - 1)$. Modification 2 switches $k^3\cdot 2^m$ voters with utility from $0$ to $1$. Therefore, the effect of the modification to the group $V$ of $\ell / k$ voters with the minimum utilities is as follows. After modification 1, at most $k\cdot (2^m - 1)$ voters in $V$ with utility $\topp$ are changed by voters with a lower utility, leading to a decrease of the total utility of at most $k^2 \cdot (2^m-1)$. Then, after modification 2, $k^3\cdot 2^m$ voters with utility $0$ are switched to utility 1, implying an increase to the expected utility of $k^3\cdot 2^m$. Therefore, the new expected utility after the modifications is at least $\ell^2 \cdot n /k + (k - 1) \cdot k^2 \cdot 2^m > \ell^2 /k \cdot \sum_{S\subseteq M} x_S $. 
\end{enumerate}

Now we round the fractional instance $\vec X_n'$ into integer instance $\vec X_n$. Following the spirit of the intermediate value theorem, there exists a rounding such that $\|\vec X_n\|_1 = n$ and for every $S$ $x_S' - 1 \le x_S \le x_S' + 1$. Therefore, for each $\ell$,  $\sum_{t = 0}^{\topp} \sum_{i \in S, |S \cap W| = t} |x_S - x_S'| \le 2^{m - \ell}$, which is a constant in our setting. Then we show that $\vec X_n$ is also an inner point of $\mathcal{H}'$. 
\begin{enumerate}
    \item For constraint 1, the left-hand side difference between $x_S'$ and $x_S$ is at most $\sum_{t = 0}^{\topp} \sum_{i \in S, |S \cap W| = t} 1$, while $x_S'$ exceeds the threshold of $k\cdot \sum_{t = 0}^{\topp} \sum_{L_W \subseteq S, |S \cap W| = t} 1$. 
    \begin{align*}
    \sum_{t = 0}^{\topp} \sum_{L_W \subseteq S, |S \cap W| = t} x_S \!\ge \! &\ 
        \sum_{t = 0}^{\topp} \sum_{L_W \subseteq S, |S \cap W| = t} x_S' - \sum_{t = 0}^{\topp} \sum_{L_W \subseteq S, |S| = t + \ell} |x_S' - x_S|\\ 
        \ge&\ \ell /k \cdot \sum_{S\subseteq M} x_S + (k - 1)\cdot \sum_{t = 0}^{\topp} \sum_{L_W \subseteq S,|S \cap W| = t} 1\\
        >&\  \ell /k \cdot \sum_{S\subseteq M} x_S.
    \end{align*}
    \item For constraint 2, the difference of the left-hand side  between $x_S'$ and $x_S$ is at most $2^{m - \ell}$, while the gap between the left-hand side with respect to $x_S'$ and the right-hand side is $\Theta(n)$. Therefore, for all sufficiently large $n$, constraint 2 is satisfied, and the inequality is strict. 
    \item For constraint 3, 
    \begin{align*}
&\ \sum_{t = 1}^{\topp} \sum_{L_W \subseteq S, |S \cap W| = t} t\cdot x_S - \left(\sum_{t = 0}^{\topp} \sum_{L_W \subseteq S, |S \cap W| = t} x_S - \ell \cdot n/ k \right)\cdot \topp \\
=&\ -\sum_{t = 1}^{\topp} \sum_{L_W \subseteq S, |S \cap W| = t} (\topp - t)\cdot x_S +\topp \cdot \ell \cdot n/k\\
=&\ -\sum_{t = 1}^{\topp} \sum_{L_W \subseteq S, |S \cap W| = t} (\topp - t)\cdot \left(x_S' - (x_S' - x_S)\right) +\topp \cdot \ell \cdot n/k\\
\ge &\ -\sum_{t = 1}^{\topp} \sum_{L_W \subseteq S, |S \cap W| = t} (\topp - t)\cdot x_S+\topp \cdot \ell \cdot n/k - (\topp-1)\cdot 2^{m - \ell}\\
=&\ \sum_{t = 1}^{\topp} \sum_{L_W \subseteq S, |S \cap W| = t} t \!\cdot\! x_S' \!-\! \left(\sum_{t = 0}^{\topp} \sum_{L_W \subseteq S, |S \cap W| = t} x_S' \!-\! \ell \!\cdot\! n/ k\right)\cdot \topp \!-\! (\topp\!-\!1)\!\cdot\! 2^{m - \ell}\\
\ge &\ \ell^2 \cdot n /k + (k - 1) \cdot k^2 \cdot 2^m - (\topp-1) \cdot 2^{m - \ell}\\
>&\  \ell^2 /k \cdot \sum_{S\subseteq M} x_S . 
\end{align*}
The first inequality comes from that $t \ge 1$ so that $\topp - t \le \topp - 1$ and $\sum_{t = 0}^{\topp} \sum_{i \in S, |S| = t + \ell} |x_S - x_S'| \le 2^{m  - \ell}$.  The second inequality comes from the satisfaction of constraint 3 of $x_S'$. 
\end{enumerate}
Therefore, we have shown that $\vec X_n$ is an inner point of $\mathcal{H}'$ for all sufficiently large $n$. 
\end{proof}

\begin{claim}
    $\mathcal{H}'_{\le 0}$ is full-dimensional. $\dim(\mathcal{H}'_{\le 0}) = 2^m$.
\end{claim}
\begin{proof}
Since $\mathcal{H}' = \mathcal{H}'_{\le 0}$, $\vec X_n$ is an inner point of $\mathcal{H}'_{\le 0}$. Therefore, $\mathcal{H}'_{\le 0}$ is full-dimensional.
\end{proof}

With all three claims, we apply Corollary~\ref{coro:pmv1} and show that $\Pr[\vec X_{\pi} \in \mathcal{H}] = \Theta(1)$. 

Consequently, the likelihood that $W$ fails AJR on $\ell$-candidate set $L$ with $p < \sqrt[\ell]{\frac{\ell}{k}}$ is $o(1)$, on $L$ with $L\cap W \neq \emptyset$ is $o(1)$, and on all the other $L$ is $1 - \Theta(1)$. Therefore, the likelihood of $W$ providing AJR is $\Theta(1)$, which completes the proof.

\section{Conclusion and Future Directions}
\label{sect:conclusion}

In this paper, we give a complete characterization of the substantial likelihood of an AJR committee when instances are sampled from the Erd\H{o}s-R\'enyi bipartite model.
Our results not only theoretically verify the previous empirical observations, such as~\citet{brill2022individual,brill2025individual}, but also give more detailed structural insights.

One natural future direction is to see if our results extend to other random models that are more general and/or more practical.
The Erd\H{o}s--R\'enyi bipartite model, albeit mathematically simple and natural, can be inaccurate when describing the real-world election instances, as the events that voters approve candidates can be dependent and with different probabilities.
Given that analyzing the seemingly simplest Erd\H{o}s-R\'enyi model is already quite technically involved, it is expected that studying the likelihood of AJR committees' existence on other more general/realistic models is even more challenging.
Nevertheless, our paper moves a first step towards this goal, and our observations and techniques can be potentially useful for future studies along this direction.
%Some more interesting and meaningful directions can be further explored. After we know the existence of AJR, is it possible to design a rule that outputs AJR committees when an AJR committee exists?
%Moving forward, future research could also expand on this work by investigating AJR on some other probabilistic models, e.g., the approval ballots satisfy some probability distributions. 
Additionally, exploring the existence of winning committees satisfying other ideal properties, such as core stability, might provide valuable insights into tackling the challenge in the classical multi-winner approval voting, since it is still an open problem whether the core stability can always be satisfied in multi-winner approval voting. 
%Therefore, one interesting future direction is to see if it is possible to extend the result of \citet{Xia2025linear} to more general sampling models.
%Furthermore, testing the existence of an AJR committee in a real dataset could enhance our understanding of designing the AJR rule in practice.

In addition, this work can be extended along the direction of proportionality degree.
Instead of focusing on the particular choice $f(\ell)=\ell$ under the AJR definition, we can study the distribution of the proportionality degree, and this is interesting even under the simplest Erd\H{o}s-R\'enyi model.
For example, how does the average satisfaction for an $\ell$-cohesive group increase above $f(\ell)=\ell$ as $p$ increases in $(p_2^\ast,1]$ before uniformly equals $k$ at $p=1$?
Also, given that PAV guarantees $f(\ell)=\ell-1$, is this tight under the Erd\H{o}s-R\'enyi model?
In particular, between $p_1^\ast$ and $p_2^\ast$ where we have shown $f(\ell)=\ell$ is unlikely to hold, are there some values of $p$ where the minimum average satisfaction is concentrated around $\ell-\frac12$?
For one step towards the last question, our Proposition~\ref{prop:ell>=2} and Lemma~\ref{lem:main} imply that a cohesive witness to the violation of $f(\ell)=\ell$ can only be a $1$-cohesive group (for all values of $p$).
Thus, we only need to analyze $1$-cohesive groups for this question with $\ell=1$.

We have also mentioned after Lemma~\ref{lem:main} and in Proposition~\ref{prop:ell>=2} that $1$-cohesive groups are the only barriers that prevent the existence of AJR committees.
It is interesting to see how broadly this observation can be applied.
Is it a feature that is unique for Erd\H{o}s-R\'enyi model?

Finally, our results can be further refined in the following aspects.
We can consider the extension to the general setting where $m$ and $k$ are not necessarily constants.
Our results can also be refined by studying the ``rates'' in the probabilities at the two transition points.
For example, how fast does the probability change when $p$ is approaching each of the two transition points?

\section*{Acknowledgments}
The research of Biaoshuai Tao is supported by the National Natural Science Foundation of China (No. 62472271).
The research of Lirong Xia is supported by NSF 2450124, 2517733, and 2518373.

The authors gratefully thank the anonymous reviewers from SODA'26, FOCS'25, EC'25, and WINE'24 for their valuable suggestions.

\bibliographystyle{plainnat}
\bibliography{ref}

\newpage
\appendix
\section{Proofs for the First Two Parts of Lemma~\ref{lem:main}}
\label{append:lemproof}
\subsection{Proof for the First Part: the Case $\utopp>\ell$}
\label{sect:prop1}
We first deal with the easy cases where $p < \frac1k$ and where $p=1$.
For the former, we formalize the analysis in the paragraph {\em Identifying the threshold for the existence of cohesive groups} before Lemma~\ref{lem:main} into the following lemma.

\begin{lemma}
\label{lemma: p < 1/k}
    For every $k$, $\ell \le k$, $p < \sqrt[\ell]{\frac{\ell}{k}}$, and $n$, the probability that there exists an $\ell$-size candidate set $L$ and a group of voter $V$ such that $V$ is an $\ell$-cohesive group towards $L$ is  $o(1)$. 
\end{lemma}
\begin{proof}%[Proof of Lemma~\ref{lemma: p < 1/k}]
    For a fixed $\ell$ and a fixed group $L$ of $\ell$ candidates, the expected number of agents approving all candidates in $L$ is $np^{\ell}$, while a cohesive group towards $L$ exists if and only if at least $\ell \cdot \frac{n}{k}$ approve all candidates in $L$.  When $p < \sqrt[\ell]{\frac{\ell}{k}}$, we have $np^\ell< \ell\cdot \frac{n}{k}$. Therefore, we could apply the Hoeffding inequality as follows. Let $Z_i$ be the random variable where $Z_i = 1$ if and only if agent $i$ approves all candidates in $L$. 

\begin{align*}
    \Pr[\text{A cohesive group for $L$ exists}] = &\ \Pr\left[\sum_{i = 1}^{n} Z_i \ge \ell \cdot \frac{n}{k}\right]\\
    = &\  \Pr\left[\sum_{i = 1}^{n} Z_i - \mathbb{E}\left[\sum_{i = 1}^{n} Z_i\right] \ge \ell \cdot \frac{n}{k} - np^{\ell}\right]\\
    \le &\exp\left(-2\left(\frac{\ell}{k} - p^{\ell}\right)^2 
    n\right)\\
    = & \exp(-\Theta(n)).
\end{align*}
Then, by applying a union bound on all $\binom{m}{\ell}$ of $\ell$-size candidate sets (where $m$ is a constant), the probability that an $\ell$-cohesive group exists is at most $\binom{m}{k} \cdot \exp(-\Theta(n)) = o(1)$. 
\end{proof}

Then notice that $\sqrt[\ell]{\frac{\ell}{k}}$ is minimized at $k = 1$. Therefore, when $p < \frac{1}{k}$, $p < \sqrt[\ell]{\frac{\ell}{k}}$ holds for any $\ell = 1, 2, \cdots, k$.
In this case, probability $1-o(1)$, no cohesive group exists, in which case any winning committee automatically satisfies AJR.

For the latter case where $p=1$, any winning committee gives each voter utility $k$, in which case AJR is also satisfied.

To conclude the first part of the lemma, we will prove the following proposition.
\begin{proposition}
    \label{prop: u>1} 
    For any constant $m, k, p$, if $\utopp > \ell$ for every $\ell \le \min\{k,m-k\}$ such that $p^\ell \ge \frac{\ell}{k}$, then the likelihood that an AJR committee exists is $1 - o(1)$. 
\end{proposition}
    %\qishen{I kind of feel that this kind of proof should come out only after you give enough high-level ideas on the statements to the readers.}
\begin{proof}
    We fix the winning committee $W = \{1, 2, \ldots, k\}$, and we will prove that $W$ satisfies AJR with high probability.
    The event that $W$ fails AJR can be divided into three scenarios: (1) when $p < \sqrt[\ell]{\frac{\ell}{k}}$ (in which case $p^\ell<\frac\ell k$ and we are not guaranteed $u_\ell>\ell$), an $\ell$-cohesive group is underrepresented, (2) when $p \ge \sqrt[\ell]{\frac{\ell}{k}}$, an $\ell$-cohesive group towards $L$ with $L\cap W = \emptyset$ is underrepresented, and (3) when $p \ge \sqrt[\ell]{\frac{\ell}{k}}$, an $\ell$-cohesive group towards $L$ with $L\cap W \neq \emptyset$ is underrepresented. We show that each scenario has a likelihood $o(1)$
    to occur. Therefore, according to the union bound, the probability that $W$ fails to provide AJR is $o(1)$.

    %Firstly, we consider the probability that $W$ fails on an $\ell$-cohesive group with $p < \sqrt[\ell]{\frac{\ell}{k}}$. 
    \textbf{Case 1.} 
    Lemma~\ref{lemma: p < 1/k} directly implies that the probability that such an $\ell$-cohesive group exists when $p < \sqrt[\ell]{\frac{\ell}{k}}$ is $o(1)$. 
    % For any fixed $\ell$ and any fixed $\ell$-candidate set $L$, as shown at the beginning of Sect.~\ref{sect:lemproof}, the probability that an $\ell$-cohesive group approving $L$ exists is $\Theta(\exp(-n))$. 
    % There are at most constantly many ($\binom{m}{\ell}$) such sets $L$ and constantly many (at most $k$) such $\ell$. Therefore, applying the union bound on all such $\ell$ and $L$, the probability that $W$ fails on some $\ell$-cohesive group with  $p < \sqrt[\ell]{\frac{\ell}{k}}$ is $o(1)$. In the rest of the proof, we assume that $p \ge \sqrt[\ell]{\frac{\ell}{k}}$, and then we calculate, $\topp$ and $\utopp$ based on $\ell$. 

    %Secondly, we show that for any $\ell \le k$ and any $\ell$-candidate set $L$ with $L\cap W = \emptyset$, the probability that $W$ fails to provide AJR on an $\ell$-cohesive group towards $L$ is o(1). 
    \textbf{Case 2.} For each $t = 0, 1,\ldots, k$, let $V_t$ be the group of voters who approve every candidate in $L$ as well as exactly $t$ candidates in $W$. Since $L\cap W = \emptyset$, the expectation of $|V_t|$ is $\binom{k}{t}\cdot p^{\ell + t}\cdot (1 - p)^{k - t}\cdot n$.  
    % Let $\varepsilon = \frac{\ell}{k} - \sum_{t = 0}^{\topp - 1} \binom{k}{t} \cdot p^{\ell + t}\cdot (1 - p)^{k - t}$, which is the number of $V_t$ voters in the $\ell$-cohesive group $V$ divided by $n$. According to the definition of $\topp$, $\varepsilon > 0$ and does not depend on $n$.
    For each $t$ and any constant $\varepsilon > 0$, by applying the Hoeffding's inequality, the probability that $||V_t| - \mathbb{E}[|V_t|]| \ge \varepsilon n$ is at most $\Theta(\exp(-n))$. That is, the size of $V_t$ concentrates to its expectation with high probability. Applying the union bound on all $t = 0,1, \ldots, k$, the probability that there exists a $t$ such that $|V_t|$ deviates from its expectation by at least $\varepsilon n$ is also at most $\Theta(\exp(-n))$. %\qishen{This is the part of what I said ``why we can consider only expectation in the Lemma''.}

    Conditioned on that all $|V_t|$ concentrates to their expectations respectively, we can lower bound the average utility of a group $V$ with high probability. A group $V$ with the minimum average utility can be generated by greedily fetching all voters in $V_0$, then $V_1$, $V_2$, $\dots$ until it has at least $\ell\cdot \frac{n}{k}$ voters. For a sufficiently small constant $\varepsilon$ (which does not depend on $n$), since each $|V_t|$ does not deviate from expectation for more than $\varepsilon n$, the average utility of $V$ will be at least $u_\ell - \frac{k}{\ell\cdot n}\cdot k\cdot \varepsilon n\cdot k = u_\ell - \frac{k^3}{\ell}\cdot \varepsilon$ (at most $k$ groups of $V_t$, each group has at most $\varepsilon\cdot n$ more voters than expectation, substituting voters with an utility of at most $k$). Note that in case 2, we have $u_\ell > \ell$. This is because (1)    $p^{\ell} \ge \frac\ell k$, (2) $\ell \le k$ by definition, and (3) $L\cap W = \emptyset$ implies that $\ell \le m - k$, which together guarantees that $u_\ell > \ell$ by the statement of Proposition~\ref{prop: u>1}.
    Therefore, by taking an $\varepsilon < (u_{\ell} - \ell) \cdot \frac{\ell}{k^3}$, the probability that there exists an $\ell$-cohesive group $V$ towards $L$ such that the average utility of $V$ is strictly less than $\ell$ is bounded by the likelihood where some $|V_t|$ deviates from the expectation more than $\varepsilon\cdot n$, leading to $\Theta(\exp(-n))$. 

    %Finally, it remains to show that $W$ fails AJR on $L$ with $L\cap W \neq \emptyset$ with probability o(1). 
    \textbf{Case 3.} %Although the expected average utility  is not assumed}, 
    For $\ell$-cohesive groups towards $L$ with $L\cap W \neq \emptyset$, the following lemma guarantees that the average utilities of those groups are strictly larger than $\ell$ with high probability. 

    \begin{lemma}
    \label{lem:overlap} 
    For any constant $m, k, p$, for every $1\le h <  \ell \le k$ such that $p^\ell \ge \frac{\ell}{k}$,
    %For any $m, k, \ell > 1, h < \ell$, and $p\ge \sqrt[\ell]{\frac{\ell}{k}}$, 
    we have $u_{\ell,h}^\ast > u_{\ell-h} + h$ 
    where $u_{\ell,h}^\ast$ is the minimum expected average utility among all $\ell$-cohesive groups towards $L$ such that $|L\cap W| = h$.
\end{lemma}
    \begin{proof}

% An easy description: incident A implies incident B, thus $|V_{A,t}|<|V_{B,t}|$. While incident A takes $l\cdot \frac{n}{k}$ voters to calculate the average, which is strictly more than the number of incident B takes. Therefore, $t_A>t_B$ and $u_{\ell,h} > u_{\ell-h} + h$
% \begin{align*}
%     u_{\ell,h}^\ast & = \frac{k}{\ell\cdot n}\cdot \left( \sum_{t = h}^{\topp^\ast} t \cdot \binom{k - h}{t - h}\cdot p^{\ell + t - h}\cdot (1 - p)^{k - t}\cdot n - \topp^\ast\cdot (\ntopp^\ast - \ell\cdot n/k)\right)\\
%     &= \frac{k}{\ell\cdot n}\cdot \left( \sum_{t = 0}^{\topp^\ast-h} t \cdot \binom{k - h}{t}\cdot p^{\ell + t}\cdot (1 - p)^{k - t-h}\cdot n - (\topp^\ast-h)\cdot (\ntopp^\ast - \ell\cdot n/k)\right)+h\\
%     &> \frac{k}{\ell\cdot n}\cdot \left( \sum_{t = 0}^{t_{\ell-h}} t \!\cdot\! \binom{k}{t} \!\cdot\! p^{\ell + t - h}\cdot (1 - p)^{k - t}\cdot n - t_{\ell-h}\cdot \left(n_{\ell-h} - (\ell - h)\cdot n/k\right)\right)+h\\
%     &> .../ tobe continued
% \end{align*}

% --------------------------------------------

Consider the $\ell$-cohesive group $V_1$ with $\ell \cdot n / k$ voters. The expected number of voters that approve $L$ (with $|L\cap W| = h$) as well as exactly $t$ winners in $W$ is $\binom{k-h}{t - h}\cdot p^{\ell + t - h}\cdot (1 - p)^{k - t}\cdot n$.
Let $\topp^\ast$ be the smallest integer such that
\begin{equation*}
\ntopp^\ast := \sum_{\topp^\ast = 0}^{\topp} \binom{k-h}{t - h}\cdot p^{\ell + t - h}\cdot (1 - p)^{k - t}\cdot n \ge \ell \cdot \frac{n}{k}.
\end{equation*}
Then we have
\begin{align*}
    u_{\ell,h}^\ast & = \frac{k}{\ell\cdot n}\cdot \left( \sum_{t = h}^{\topp^\ast} t \cdot \binom{k - h}{t - h}\cdot p^{\ell + t - h}\cdot (1 - p)^{k - t}\cdot n - \topp^\ast\cdot (\ntopp^\ast - \ell\cdot n/k)\right).
\end{align*}
On the other hand, $u_{\ell - h}$ is the minimum expected average utility among all $(\ell - h)$-cohesive group supporting some $L'$ such that $L' \cap W = \emptyset$, 
\begin{align*}
    u_{\ell-h} \!=\! \frac{k}{(\ell - h)\cdot n}  \!\cdot\! \left( \sum_{t = 0}^{t_{\ell-h}} t \!\cdot\! \binom{k}{t} \!\cdot\! p^{\ell + t - h}\cdot (1 - p)^{k - t}\cdot n - t_{\ell-h}\cdot \left(n_{\ell-h} - (\ell - h)\cdot n/k\right)\right). 
\end{align*}
Next, we construct $u_\ell^\prime$ satisfying $u_{\ell-h}^\ast + h < \utopp^\prime < u_\ell^\ast $. Let $\topp^\prime$ be the smallest integer such that 
\begin{equation*}
    \ntopp^\prime = \sum_{t = h}^{\topp^\prime} \binom{k}{t - h}\cdot p^{\ell + t - 2h}\cdot (1 - p)^{k - t - h}\cdot n \ge \frac{\ell\cdot n}{k}.
\end{equation*}
First, we show that 
\begin{align*}
    \utopp^\prime \!=\! \frac{k}{\ell\cdot n}\cdot \left( \sum_{t = h}^{\topp^\prime} t \cdot \binom{k}{t - h}\cdot p^{\ell + t - 2h}\cdot (1 - p)^{k - t - h}\cdot n - \topp^\prime\cdot (\ntopp^\prime - \ell\cdot n/k)\right) < u_{\ell,h}. 
\end{align*}
To verify this, $\utopp^\prime$ can be viewed as the minimum expected average utility of an $\ell$-cohesive group $V_3$ such that the expected number of voters that approve $t$ winners is $\binom{k}{t - h}\cdot p^{\ell + t - 2h}\cdot (1 - p)^{k - t - h}\cdot n$. For each $h\le t\le t_\ell^\prime$, the number of voters that approve  $t$ winners in $V_3$ is more than that in $V_1$. In addition, both $u_{\ell,h}$ and $\utopp^\prime$ consider $\ell\cdot\frac{n}{k}$ voters. Hence, we have $\utopp^\prime < u_{\ell,h}$.

Then, we show that $\utopp^\prime > u_{\ell-h}^\ast + h$. Note that $\utopp^\prime$ can be reformulated as
\begin{align*}
    \utopp^\prime =& \frac{k}{\ell\cdot n}\cdot \left( \sum_{t = 0}^{\topp^\prime - h} (t + h) \cdot \binom{k}{t}\cdot p^{\ell + t - h}\cdot (1 - p)^{k - t}\cdot n - \topp^\prime\cdot (\ntopp^\prime - \ell\cdot n/k)\right)\\
    =& \frac{k}{\ell\cdot n}\cdot \left( \sum_{t = 0}^{\topp^\prime - h} t \cdot \binom{k}{t}\cdot p^{\ell + t - h}\cdot (1 - p)^{k - t}\cdot n - (\topp^\prime - h)\cdot (\ntopp^\prime - \ell\cdot n/k)\right) \!+\! h,
\end{align*}
where the first term can be seen as the average utility of $\ell \frac{n}{k}$ voters with lower utilities in which there are $\binom{k}{t}\cdot p^{\ell + t - h}\cdot (1 - p)^{k - t}\cdot n$ voters approving $t$ winners, which is same as $V_2$. However, $u_{\ell-h}$ only considers $(\ell-h)\cdot \frac{n}{k}$ voters with lower utilities, which implies that $u_{\ell,h}^\ast > u_{\ell-h} + h$. 
\end{proof}

    Consider a set $L$ of $\ell$ candidates with $|L\cap W| = h$. If $h = \ell$, every voter approving $L$ will have the utility of at least $\ell$, and then AJR is satisfied. When $h < \ell$, we show that the likelihood that AJR fails on an $\ell$-cohesive group approving $L$ will be at most the likelihood of AJR fails on an $(\ell-h)$-cohesive group approving $L\setminus W$, which implies $o(1)$. From $|L\cap W| = h$, since $p \ge \sqrt[\ell]{\frac{\ell}{k}} > \sqrt[\ell-h]{\frac{\ell-h}{k}}$, we have that the minimum expected average utility $u_{\ell-h}$ among all $(\ell-h)$-cohesive groups towards $L\setminus W$ is strictly larger than $\ell - h$ (by the condition in Proposition~\ref{prop: u>1} and noticing $\ell-h\leq m-k$). Then by Lemma~\ref{lem:overlap}, the minimal expected average utility of an $\ell$-cohesive group towards $L$ is strictly larger than $\ell$. Next, following a similar analysis with Case 2, we have that the likelihood that AJR fails on an $\ell$-cohesive group towards $L$ is at most $o(1)$.

    Now we apply the union bound on all the cases. Since $m, k, \ell, p$ are either constants or bounded by constants, the likelihood that AJR fails is $o(1)$. 
\end{proof}

\subsection{Proof for the Second Part: the Case $\utopp<\ell$}
\label{sect:prop2}
For the second part of the lemma, we will prove the following proposition.
\begin{proposition}
    \label{prop: u < 1} 
    For any constant $m, k, p$, if there exists $\ell$ with $\ell\leq\min\{k,m-k\}$ such that $p^\ell \ge \frac{\ell}{k}$ and $\utopp < \ell$, then the likelihood that an AJR committee exists is $o(1)$. 
    %For any constant $m, k, \ell \le k$, and $p$ with $p \ge \sqrt[\ell]{\frac{\ell}{k}}$, if $\utopp < \ell$, the likelihood that every committee $W$ fails AJR is $1 - o(1)$. 
\end{proposition}
\begin{proof}
    When $p = \sqrt[\ell]{\frac{\ell}{k}}$ for some $\ell$, $\utopp$ is guaranteed to be at least $\ell$, which is proved in the lemma below.
    
    \begin{lemma}
    \label{lem:lowerp} For any $k, \ell$, when $p = \sqrt[\ell]{\frac{\ell}{k}}$, we have $u_\ell \ge \ell$. 
\end{lemma}
\begin{proof}
    When $\ell = k$, we have $p=1$, and the average utility for any subset of voters is exactly $k=\ell$, implying that $u_\ell = \ell$. 
    %When $\ell < k$, $p=\sqrt[\ell]{\frac{\ell}{k}}$. 
    We assume $\ell<k$ from now on.
    Since $\sum_{t = 0}^{\topp} \binom{k}{t}\cdot p^{t}\cdot (1 - p)^{k - t} \leq 1$ and the equality is met if and only if $\topp=k$, we have that $\ntopp = \sum_{t = 0}^{\topp} \binom{k}{t}\cdot p^{\ell + t}\cdot (1 - p)^{k - t}\cdot n \ge p^\ell \cdot n =  \ell \cdot \frac{n}{k}$ if and only if $\topp=k$. By our definitions of $\topp$ and $n_\ell$ in Equation~(\ref{eqn:tlandnl}), we have $\topp=k$ and $n_\ell = \ell\cdot\frac{n}{k}$, and 
    \begin{align*}
        \utopp &= \frac{k}{\ell \cdot n}\cdot \left( \sum_{t = 0}^{k} t \cdot \binom{k}{t}\cdot p^{\ell + t}\cdot (1 - p)^{k - t}\cdot n\right)\\
        &= \frac{k}{\ell}\cdot p^\ell\left( \sum_{t = 1}^{k} t \cdot \binom{k}{t}\cdot p^{t}\cdot (1 - p)^{k - t}\right)\tag{the first term is $0$}\\
        &=\frac{k}\ell\cdot p^\ell\left(\sum_{t=1}^{k}k\cdot\binom{k-1}{t-1}p^t(1-p)^{k-t}\right)\tag{rewrite the combinatoral number}\\
        &=\frac{k^2}\ell\cdot p^{\ell+1}\left(\sum_{t=0}^{k-1}\binom{k-1}{t}p^t(1-p)^{k-1-t}\right)\\
        &=\sqrt[\ell]{\ell k^{\ell-1}}\tag{since $p=\sqrt[\ell]{\frac\ell k}$ and the summation is $(p+(1-p))^{k-1}=1$}\\
        &\geq \ell, \tag{since $k>\ell$ and $\ell\geq1$}
    \end{align*}
    
    %Then we consider the case where $1<\ell < k$. In the special case where $p = \sqrt[\ell]{\frac{\ell}{k}}$, the expected number of voters approving a certain $\ell$-candidate set $L$ is exactly $\ell \cdot \frac{n}{k}$. In this case, an $\ell$-cohesive group should contain all these voters, where the expected average utility is $u_p =p\cdot k = \sqrt[\ell]{\ell \cdot k^{\ell - 1}}> \ell$. We can use a similar analysis to the proof of Proposition~\ref{prop: u>1} to show that the minimal expected average utility of $V$ will be at least $u_p - \frac{k}{\ell\cdot n}\cdot k\cdot \varepsilon n\cdot k = u_\ell - \frac{k^3}{\ell}\cdot \varepsilon > \ell$ by taking a sufficiently small $\varepsilon$. 
    
which concludes the lemma.
\end{proof}
    
    To conclude Proposition~\ref{prop: u < 1}, we only need to consider the case where $p^\ell > \frac{\ell}{k}$.
    For each winning committee $W$, we fix an $\ell$-candidate set $L_W$ such that $L_W\cap W = \emptyset$.
    Then we show that $W$ is likely to fail AJR on an $\ell$-cohesive group towards $L_W$. 
    Finally, we apply the union bound on all $W$ and show that the likelihood that every $W$ fails AJR on $L_W$ is also $1 - o(1)$.

    This proof is similar to Case 2 in the proof of Proposition~\ref{prop: u>1}. For each $t = 0, 1,\ldots, k$, let $V_t$ be the group of voters who approve all candidates in $L$ as well as exactly $t$ candidates in $W$. For some sufficiently small $\varepsilon$, the probability where there exists a $t \in [\topp]$ such that $|V_t|$ deviates from its expectation by at least $\varepsilon \cdot n$ is at most $\Theta(\exp(-n))$.  
    
    Conditioned on that no $|V_t|$ deviates much, we consider the average utility of the group by picking voters from $V_0, V_1, \dots$ until it has at least $\ell\cdot \frac{n}{k}$ voters. For sufficiently small $\varepsilon$, since each $|V_t|$ does not deviate by more than $\varepsilon n$, the average utility of $V$ will be at most $u_\ell+\frac{k}{\ell\cdot n}\cdot k\cdot \varepsilon n\cdot k = u_\ell + \frac{k^3}{\ell}\cdot \varepsilon$ (at most $k$ groups of $V_t$, each group has at most $\varepsilon\cdot n$ less voters that are substituted by voters with a utility of at most $k$).
    Therefore, for sufficiently small $\varepsilon$, the probability that there exists an $\ell$-cohesive group $V$ towards $L_W$ with an average utility strictly less than $\ell$ is at least $1 - o(1)$. %Finally, by applying a union bound on all $W$ ($\binom{m}{k}$ in total), the probability that every $W$ satisfies AJR is at most $o(1)$, which completes the proof. 
    Thus, the probability that an arbitrary given winning committee $W$ satisfies AJR is $o(1)$.
    By a union bound on all $\binom{m}{k}$ possible winning committees (and notice that $\binom{m}{k}$ is a constant), the probability that there exists an AJR committee is $o(1)$.
\end{proof}

\section{Missing Proofs in Sect.~\ref{sect:proofofmaintheorem}}
\label{append:proofofmaintheorem}
\subsection{Proof of Proposition~\ref{prop:U}}
\label{append:propU}
    Note that $u_\ell=U(t_\ell)$. Let $C_t=\binom{k}{t} p^{\ell+t}(1-p)^{k-t}$. For all $T<t_\ell$,
    \begin{align*}
        U(T) &= \frac{k}{\ell}\left( \sum_{t=0}^{T} tC_t-T\left( \sum_{t=0}^{T} C_t -\frac{\ell}{k}\right) \right)\\
        &= \frac{k}{\ell}\left( \sum_{t=0}^{T} tC_t+\sum_{t=T+1}^{t_l}tC_t-\sum_{t=T+1}^{t_\ell}tC_t-T\left( \sum_{t=0}^{T} C_t -\frac{\ell}{k}\right) \right)\\
        &< \frac{k}{\ell}\left( \sum_{t=0}^{T} tC_t+\sum_{t=T+1}^{t_l}tC_t-T\sum_{t=T+1}^{t_\ell}C_t-T\left( \sum_{t=0}^{T} C_t -\frac{\ell}{k}\right) \right)\\
        &=U(t_\ell).
    \end{align*}
    For all $T>t_\ell$,
    \begin{align*}
        U(T) &= \frac{k}{\ell}\left( \sum_{t=0}^{T} tC_t-T\left( \sum_{t=0}^{T} C_t -\frac{\ell}{k}\right) \right)\\
        &= \frac{k}{\ell}\left( \sum_{t=0}^{T} tC_t-\sum_{t=t_\ell+1}^{T}tC_t+\sum_{t=t_\ell+1}^{T}tC_t-T\left( \sum_{t=0}^{T} C_t -\frac{\ell}{k}\right) \right)\\
        &< \frac{k}{\ell}\left( \sum_{t=0}^{T} tC_t-\sum_{t=t_\ell+1}^{T}tC_t+T\sum_{t=t_\ell+1}^{T}C_t-T\left( \sum_{t=0}^{T} C_t -\frac{\ell}{k}\right) \right)\\
        &=U(t_\ell).
    \end{align*}
    We conclude the proposition.

\subsection{Proof of Proposition~\ref{conjec:complex_ineq}}
\label{append:conjecture2}
The proof consists of four parts:
\begin{enumerate}
    \item we first prove that the inequality holds for small $k$, with $k\leq\frac{\sqrt5+1}2\ell$,
    \item next, we prove that the inequality holds for large $k$, with $k\geq29\ell$,
    \item for the middle regime $\frac{\sqrt5+1}2\ell\leq k\leq29\ell$, we show that the inequality holds for sufficiently large $\ell$, with $\ell\geq3947$,
    \item finally, the remaining values of $(k,\ell)$ with $\ell<3947$ and $\frac{\sqrt5+1}2\ell\leq k\leq29\ell$ are checked by computer programs.
\end{enumerate}

\begin{proposition} \label{prop::conjecture2_smallK}
    $\binom{k}{\ell}\left(1-\sqrt[\ell]{\frac{\ell}{k}}\right)^{k-\ell}<1$ holds when $k\leq \frac{\sqrt{5}+1}{2}\ell$.
\end{proposition}

\begin{proof}
    Since $\ln(1-x)< -x$ for $x\in (0,1)$, substituting $x=\frac{1}{\ell}\ln\frac{k}{\ell}$ into the expression, we have
    \begin{equation*}
        \ln\left(1-\frac{1}{\ell}\ln\frac{k}{\ell}\right) < -\frac{1}{\ell}\ln\frac{k}{\ell},
    \end{equation*}
    where the inequality holds when $k\leq \ell e^\ell$. Exponentiating both sides, we have
    \begin{equation*}
        1-\frac{1}{\ell}\ln\frac{k}{\ell} < e^{-\frac{1}{\ell}\ln\frac{k}{\ell}} = \sqrt[\ell]{\frac{\ell}{k}}.
    \end{equation*}
    Combining with $\binom{k}{\ell}\leq k^{k-\ell}$, we further have
    \begin{equation*}
        \binom{k}{\ell}\left(1-\sqrt[\ell]{\frac{\ell}{k}}\right)^{k-l} < k^{k-\ell} \cdot \left(\frac{1}{\ell}\ln\frac{k}{\ell}\right)^{k-\ell}.
    \end{equation*}
    Since $\ln(1+x) < x$ for $x > 0$, substituting $x=\frac{k}{\ell}-1$ into the expression, we have
    \begin{equation*}
        k^{k-\ell} \cdot \left(\frac{1}{\ell}\ln\frac{k}{\ell}\right)^{k-\ell} < k^{k-\ell} \cdot \left(\frac{1}{\ell}\cdot \frac{k-\ell}{\ell}\right)^{k-\ell} = \left(\frac{k(k-\ell)}{\ell^2} \right)^{k-\ell}.
    \end{equation*}
    When $k\leq \frac{\sqrt{5}+1}{2}\ell$, we have
    \begin{equation*}
        \frac{k(k-\ell)}{\ell^2} \le \frac{\sqrt{5}+1}{2}\left(\frac{\sqrt{5}+1}{2}-1\right) = 1,
    \end{equation*}
    which completes the proof.
\end{proof}

\begin{proposition}\label{prop::conjecture2_largeK}
    $\binom{k}{\ell}\left(1-\sqrt[\ell]{\frac{\ell}{k}}\right)^{k-l}<1$ holds when $k\geq 29\ell$.
\end{proposition}

\begin{proof}
    Since $\ln(1-x)< -x$ for $x\in (0,1)$, substituting $x=\left(\frac{\ell}{k}\right)^\frac{1}{\ell}$ into the expression, we have
    \begin{equation*}
        \ln\left(1-\left(\frac{\ell}{k}\right)^\frac{1}{\ell}\right)< - \left(\frac{\ell}{k}\right)^\frac{1}{\ell}.
    \end{equation*}
    Multiplying both sides by $(k-\ell)$ and exponentiating both sides, we further have
    \begin{equation*}
        \left(1-\left(\frac{\ell}{k}\right)^\frac{1}{\ell}\right)^{k-\ell}< \exp\left(-(k-\ell)\left(\frac{\ell}{k}\right)^\frac{1}{\ell}\right).
    \end{equation*}
    Since $\binom{k}{\ell} \le \left(\frac{ek}{\ell}\right)^\ell$, we only need to show
    \begin{equation*}
        \left(\frac{ek}{\ell}\right)^\ell \cdot\exp\left(-(k-\ell)\left(\frac{\ell}{k}\right)^\frac{1}{\ell}\right)\le 1,
    \end{equation*}
    which is equivalent to
    \begin{equation*}
        \ell(1+\ln k -\ln \ell) -k\left(\frac{\ell}{k}\right)^\frac{1}{\ell}+\ell\left(\frac{\ell}{k}\right)^\frac{1}{\ell}\le 0.
    \end{equation*}
    Since $\ell\left(\frac{\ell}{k}\right)^\frac{1}{\ell} < \ell$ and let $t=\frac{k}{\ell}$, we have
    \begin{equation*}
        \ell(1+\ln k -\ln \ell) -k\left(\frac{\ell}{k}\right)^\frac{1}{\ell}+\ell\left(\frac{\ell}{k}\right)^\frac{1}{\ell} < 2\ell +\ell \ln t - \ell\cdot t^{1-\frac{1}{\ell}}\le (2+\ln t - t^{\frac{1}{2}})\ell,
    \end{equation*}
    and it suffices to show $(2+\ln t - t^{\frac{1}{2}})\ell<0$.
    
    Since $\left(2+\ln t - t^{\frac{1}{2}}\right)$ is monotonically decreasing with respect to $t\ge 4$ and $2+\ln 29 - 29^{\frac{1}{2}}<0$, we have $\left(2+\ln t - t^{\frac{1}{2}}\right)\ell < 0$ when $t \ge 29$, i.e., $k\ge 29\ell$.
\end{proof}

\begin{proposition}
    $\binom{k}{\ell}\left(1-\sqrt[\ell]{\frac{\ell}{k}}\right)^{k-\ell}<1$ holds when $\frac{\sqrt{5}+1}{2}\ell< k< 29\ell$ and $\ell \ge 3947$.
\end{proposition}

\begin{proof}
    First, we can use a similar analysis as Proposition~\ref{prop::conjecture2_smallK} and Proposition~\ref{prop::conjecture2_largeK}, we have $\binom{k}{\ell}\le \left(\frac{ek}{\ell}\right)^\ell$ and $\left(1-\sqrt[\ell]{\frac{\ell}{k}}\right)^{k-\ell} < \left(\frac{1}{\ell}\ln\frac{k}{\ell}\right)^{k-\ell}$. Thus,
    \begin{equation*}
        \binom{k}{\ell}\left(1-\sqrt[\ell]{\frac{\ell}{k}}\right)^{k-l} < \left(\frac{ek}{\ell}\right)^\ell \cdot \left(\frac{1}{\ell}\ln\frac{k}{\ell}\right)^{k-\ell} = \left(\left(\frac{ek}{\ell}\right) \cdot \left(\frac{1}{\ell}\ln\frac{k}{\ell}\right)^{\frac{k}{\ell}-1}\right)^\ell.
    \end{equation*}
    Letting $t=\frac{k}{\ell}$, we only need to show
    \begin{equation*}
        et \left(\frac{1}{\ell}\ln t\right)^{t-1} < 1.
    \end{equation*}
    Since $\ell\ge 3947$ and $t\le 29$, we have $\frac{1}{\ell}\ln t < 1$. Hence, 
    \begin{equation*}
        et \left(\frac{1}{\ell}\ln t\right)^{t-1} < 29e \left(\frac{1}{3947}\ln 29\right)^{\frac{\sqrt{5}-1}{2}}<1.
    \end{equation*}
\end{proof}

For the remaining case, we can use numerical analysis to show that $\binom{k}{\ell}\left(1-\sqrt[\ell]{\frac{\ell}{k}}\right)^{k-\ell}<1$ holds for every $\ell < 3947$ and $\frac{\sqrt{5}+1}{2}\ell< k< 29\ell$.
The codes (in C++) for this checking are presented below.

\begin{lstlisting}[style=cppstyle]
#include<iostream>
#include<algorithm>
#include<cmath>

using namespace std;

int main()
{
   for (int ell = 2; ell < 3948; ell++)
   {
      for (int k = (1.5 * ell); k < ell * 29 + 1; k++)
      {
         double result = 1;
         for (int i = 0; i < k - ell; i++)
         {
            // one term in the combinatorial number
            result *= 1.0 * (k - i) / (i + 1); 
            // one term for (1 - sqrt[ell]{ell / k})^{k - ell}
            result *= 1 - pow(1.0 * ell / k , 1.0 / ell); 
         }
         if (result >= 1)
         {
            cout << "Error: " << ell << " " << k << " " << result << endl;
            return 0;
         }
      }
      cout << ell << endl;
   }
   return 0;
}
\end{lstlisting}

\iffalse
\begin{minted}{c} 
#include<iostream>
#include<algorithm>
#include<cmath>

using namespace std;

int main()
{
   for (int ell = 2; ell < 3948; ell++)
   {
      for (int k = (1.5 * ell); k < ell * 29 + 1; k++)
      {
         double result = 1;
         for (int i = 0; i < k - ell; i++)
         {
            // one term in the combinatorial number
            result *= 1.0 * (k - i) / (i + 1); 
            // one term for (1 - sqrt[ell]{ell / k})^{k - ell}
            result *= 1 - pow(1.0 * ell / k , 1.0 / ell); 
         }
         if (result >= 1)
         {
            cout << "Error: " << ell << " " << k << " " << result << endl;
            return 0;
         }
      }
      cout << ell << endl;
   }
   return 0;
}
\end{minted}
\fi

\subsection{Detailed Calculations of $\frac{\partial U(T)}{\partial p}$ in Claim~\ref{prop:decreasing}}
\label{append:partialcalculation}
\begin{align*}
            \allowdisplaybreaks
            \frac{\partial U(T)}{\partial p} =& -k\left( \sum_{t=0}^{T} (T-t)\binom{k}{t} p^{t}(1-p)^{k-t}\right)\\
            &-k\left(\sum_{t=1}^{T} (T-t)k\binom{k-1}{t-1}p^{t}(1-p)^{k-t} -\sum_{t=0}^{T} (T-t)k\binom{k-1}{t}p^{1+t}(1-p)^{k-t-1} \right)\tag{differentiate and rewrite the combinatorial numbers}\\
            =& -k\left( \sum_{t=0}^{T} (T-t)\binom{k}{t} p^{t}(1-p)^{k-t}\right)\\
            &-k\left(\sum_{t=0}^{T-1} (T+1-t)k\binom{k-1}{t}p^{t+1}(1-p)^{k-t-1} -\sum_{t=0}^{T} (T-t)k\binom{k-1}{t}p^{1+t}(1-p)^{k-t-1} \right)\tag{rewrite the second summation from $0$ to $T-1$}\\
            =& -k\left( \sum_{t=0}^{T-1} (T-t)\binom{k}{t} p^{t}(1-p)^{k-t}-k\sum_{t=0}^{T-1} \binom{k-1}{t}p^{1+t}(1-p)^{k-t-1}  \right)\tag{combine the 2nd and 3rd summations; noting that the last terms in the 1st and 3rd summations are $0$}\\
            =& -k\left( \sum_{t=0}^{T-1} (T-t)\binom{k}{t} p^{t}(1-p)^{k-t}-\sum_{t=0}^{T-1} (k-t)\binom{k}{t}p^{1+t}(1-p)^{k-t-1}  \right)\\
            =& -k\left( \sum_{t=0}^{T-1} \binom{k}{t} p^{t}(1-p)^{k-t-1}\left[(T-t)(1-p)-(k-t)p\right] \right).
        \end{align*}

\subsection{The Remaining Part for the Proof of Claim~\ref{prop: F(p_1)<0}}
\label{append:propF(p_1)<0}
    The formal proof for 
    $$\phi(k,p)=k^2p^2(1-p)^{k-2}(1+(k-2)p)>1$$
    goes into two parts.
    For the first part, we have used a computer program to check that $\phi(k,p)>1$ holds for all $k=3,4,\ldots,1000$ and $p$ given by (\ref{eqn:partial0-p}). The codes are attached to the end of this proof.
    Now, it remains to consider the second part with $k>1000$.

    To characterize the approximation $p\approx\frac{\sqrt{2}}{k}$, we compute
    $$\frac{p}{\frac{\sqrt{2}}{k}}=\frac{1}{\frac{\sqrt2}k+\sqrt{1-\frac1k}}\in\left(\frac1{\frac{\sqrt2}{k}+1},\frac1{\sqrt{1-\frac1k}}\right)\subseteq\left(\frac1{\frac{\sqrt2}{1000}+1},\frac1{\sqrt{1-\frac1{1000}}}\right)\subseteq(0.998,1.001).$$
    Therefore,
    $$0.998\cdot \frac{\sqrt{2}}{k}\leq p\leq 1.001\cdot\frac{\sqrt{2}}{k}. $$
    Substituting this into $\phi$, we have
    \begin{align*}
        \phi(k,p)&=k^2p^2(1-p)^{k-2}(1+(k-2)p)\\
        &\geq k^2\left(0.998\cdot\frac{\sqrt{2}}{k}\right)^2\left(1-1.001\cdot\frac{\sqrt{2}}{k}\right)^{k-2}\left(1+(k-2)\cdot0.998\cdot \frac{\sqrt{2}}{k}\right)\\
        &>\left(0.998\cdot\sqrt2\right)^2\left(1-1.001\cdot\frac{\sqrt{2}}{k}\right)^{k}\left(1+0.998\cdot\sqrt2-2\cdot0.998\cdot \frac{\sqrt{2}}{k}\right)\tag{since $1-1.001\cdot\frac{\sqrt{2}}{k}<1$}\\
        &>\left(0.998\cdot\sqrt2\right)^2\left(1-1.001\cdot\frac{\sqrt{2}}{1000}\right)^{1000}\left(1+0.998\cdot\sqrt2-2\cdot0.998\cdot \frac{\sqrt{2}}{1000}\right)\tag{since $k>1000$ and the function $f(k)=(1-1/k)^k$ is increasing in $k$}\\
        &>1.16,\tag{computed by a calculator}
    \end{align*}
    which concludes the claim.

The codes (in C++) below check that $\phi(k,p)>1$ holds for all $k=3,4,\ldots,1000$ and $p$ given by (\ref{eqn:partial0-p}).

\begin{lstlisting}[style=cppstyle] 
#include<iostream>
#include<algorithm>
#include<cmath>

using namespace std;

int main()
{
   for (int k = 3; k <= 1000; k++)
   {
      double p = 1.0 / (1 + pow(k * (k - 1) / 2, 0.5));
      double E = k * k * p * p * pow(1 - p, k - 2) * (1 + (k - 2) * p);
      cout << E << endl;
      if (E <= 1)
      {
         cout << "Error: " << k << " " << E << endl;
         return 0;	
      }	
   }
   return 0;
}
\end{lstlisting} 

\end{document}